\documentclass[11pt]{article}
\newcommand{\blind}{1}
\RequirePackage[colorlinks,citecolor=blue,urlcolor=blue]{hyperref}
\RequirePackage{graphicx}
\usepackage{amsfonts}
\usepackage{bbm}
\usepackage{natbib}
\usepackage{mathtools}
\usepackage[left=1in,right=1in,top=1in,bottom=1in]{geometry}
\usepackage{subcaption} 
\usepackage{graphicx}
\usepackage{pgfplots}
\usepackage[all]{nowidow}
\usepackage[utf8]{inputenc}
\usepackage{tikz}
\usepackage{multicol}
\usepackage{algpseudocode,algorithm,algorithmicx}

\usepackage{scalerel}
\usepackage{natbib}
\usepackage[normalem]{ulem}

\usepackage[english]{babel}

\usepackage{mathdots}
\usepackage{yhmath}
\usepackage{cancel}
\usepackage{color}
\usepackage{array}
\usepackage{multirow}
\usepackage{enumerate}
\usepackage{gensymb}
\usepackage{tabularx}
\usepackage{booktabs}
\usetikzlibrary{fadings}
\usetikzlibrary{patterns}
\usetikzlibrary{shadows.blur}

\usepackage{enumerate}
\usepackage[algo2e]{algorithm2e} \usepackage{easyReview}
\usepackage{amssymb}
\usepackage{amsmath}    
\usepackage{dirtytalk}
\usepackage{amsthm}

\newcommand{\cH}{\mathcal{H}}
\newcommand{\cL}{\mathcal{L}}

\newcommand{\indep}{\perp \!\!\! \perp}
\newcommand{\bE}{\mathbb{E}}
\newcommand{\bP}{\mathbb{P}}
\newcommand{\cN}{\mathcal{N}}

\theoremstyle{plain}

\newtheorem{theorem}{Theorem}[section]
\newtheorem{lemma}[theorem]{Lemma}
\newtheorem{proposition}{Proposition}
\theoremstyle{definition}
\newtheorem{definition}[theorem]{Definition}
\newtheorem{assumption}{Assumption}
\newtheorem*{example}{Example}

\newtheorem{remark}{Remark}
\theoremstyle{remark}

\usepackage[nodisplayskipstretch]{setspace}

\begin{document}
\date{}

\def\spacingset#1{\renewcommand{\baselinestretch}%
{#1}\small\normalsize} \spacingset{1}


\if1\blind
{
  \title{ Resampling-free Inference for Time Series via RKHS Embedding}
  \author{Deep Ghoshal \hspace{.2cm}\\
    Department of Statistics, University of Illinois at Urbana-Champaign\\
    and \\
    Xiaofeng Shao \\
   Department of Statistics and Data Science,
 Washington University in St Louis}
  \maketitle
} \fi

\if0\blind
{
  \bigskip
  \bigskip
  \bigskip
  \begin{center}
    {\LARGE\bf Resampling-free Inference for Time Series }
\end{center}
  \medskip
} \fi

\bigskip
\begin{abstract}
In this article, we study nonparametric inference problems in the context of multivariate or functional time series, including testing for goodness-of-fit, the presence of a change point in the marginal distribution, and the independence of two time series, among others. Most methodologies available in the existing literature address these problems by employing a bandwidth-dependent bootstrap or subsampling approach, which can be computationally expensive and/or sensitive to the choice of bandwidth.  To address these limitations, we propose a novel class of  kernel-based tests by embedding the data into a reproducing kernel Hilbert space, and construct test statistics using sample splitting, projection, and self-normalization (SN) techniques. Through a new conditioning technique, we demonstrate that our test statistics have pivotal limiting null distributions under strong mixing and mild moment assumptions. We also analyze the limiting power of our tests under local alternatives. Finally, we showcase the superior size accuracy and computational efficiency of our methods as compared to some existing ones.
\end{abstract}

\noindent%

\spacingset{1} 
\section{Introduction}
\label{sec:intro}
Hypothesis testing is a ubiquitous problem in time series analysis. It helps us validate intuitions that have been used to model a particular dataset. Examples of such commonly faced inferential problems include testing whether a given series has a pre-specified mean, there is a change point in the mean of the series, it has Gaussian marginal distribution, a change point is present in the marginal distribution, and it is independent of another observed time series etc. While the first two problems can be categorized into testing for a finite dimensional parameter, all others are concerned with inference for an infinite-dimensional parameter. There has been a vast amount of work on inference in time series and some traditional nonparametric
approaches include block bootstrapping (\cite{kunsch1989jackknife}, \cite{liu1992moving}), subsampling (\cite{politis1994large}), and blockwise empirical likelihood
(\cite{kitamura1997empirical}) among others. However, these methods share a couple of common drawbacks. Firstly, the resampling methods most often turn out to be computationally expensive in large scale applications. Secondly, the finite sample performance of these methods is quite sensitive to the choice of the block-length involved and the optimal choice of the block size is typically a difficult problem. The existing block size selection methods may involve another ad hoc user-chosen number and their performance may be suboptimal, see Chapter 7 in  \cite{lahiri2003resampling}.\\
\indent In this article, we develop a novel SN-based approach towards inference for infinite dimensional parameters. 
Self-normalization (SN, hereafter) was developed by \cite{shao2010self} for the inference of a finite dimensional parameter of a stationary time series, extending the bandwidth-free inference first introduced by \cite{kiefer2000simple} and \cite{lobato2001testing} in the context of time series regression and white
noise testing, respectively. For a detailed review of the use of self-normalization for the inference of a low-dimensional parameter in time series, we refer the reader to \cite{shao2015self}. In the recent works of \cite{wang2020hypothesis}, \cite{wang2022inference} etc., the authors have extended the SN idea to the inference for high dimensional parameters. However, it seems that the use of SN in testing for an infinite dimensional parameter is less explored. \\

\indent To illustrate the use of our new test, we consider the nonparametric problem of testing for goodness-of-fit, the presence of a change point in the marginal distribution,  and independence of two time series. The key methodological innovation is to embed the original series into a reproducing kernel Hilbert space (RKHS) and convert the aforementioned problems into the inference for a parameter of this RKHS-valued series. We then split the sample into training and testing parts. Based on the training sample, we estimate a projection direction in the RKHS, followed by projecting the data points in the testing sample along this estimated function to form a univariate sequence.  Finally, we use this one-dimensional process to form our self-normalized statistics. On the theory front, we first demonstrate how the original testing problem is equivalent to testing for some finite-dimensional parameter of this one-dimensional process. To obtain the limiting distributions, we adopt a new conditioning technique to establish that our test statistics have pivotal asymptotic null distributions and derive their limiting power under a sequence of local alternatives. Thus, we achieve the goal of doing resampling-free inference for infinite-dimensional parameters in time series, which are computationally much more efficient than the resampling-based counterparts. In addition, our method is applicable to vector time series,   functional time series in Hilbert space, and object-valued time series, as we demonstrate in methodology description and simulation studies. \\

\indent Our work is partially inspired by \cite{kim2024dimension}, which used a sample splitting approach to do dimension-agnostic inference for iid (independent and identically distributed) data.  In comparison, their work is not tailored to incorporate the underlying temporal dependence in a time series. Moreover, for the inference of an infinite-dimensional parameter, their work is only limited to a class of degenerate U-statistics, which includes goodness-of-fit testing, whereas we propose hypothesis tests for a significantly broader class of problems,  including change point testing for marginal distribution. It is worthwhile to mention that \cite{gao2023dimension} and \cite{zhang2024another} have recently adopted the sample splitting and self-normalization approach to inference in time series data. However, both of those works focus on potentially high-dimensional but still finite-dimensional  parameters. Very recently, \cite{Zhang27032025} have proposed a methodology based on sample-splitting and self-normalization to test for functional parameters in time series data. While there are some similarities between their methodology and ours, a main advantage of our approach is that due to the use of kernels, we are able to handle Euclidean data as well as functional or non-Euclidean data in a unified fashion, whereas the method of \cite{Zhang27032025} seems tailored to Euclidean time series. 
However, it should be noted that they are able to do inference for infinite-dimensional parameters like the marginal quantile function, the spectral distribution function etc., whereas our RKHS-based approach helps us to tackle problem like testing independence between two time series. So the scope of applicabilities for these two papers are different.\\

\indent The remainder of the paper is organized as follows. We present a brief review of RKHS in  Section~\ref{subsection:review_rkhs}. Then we present our methodology for goodness-of-fit testing in Section~\ref{sec:method_GOF}, for change point testing in Section~\ref{sec:method_CPD} and for the independence testing between two time series in Section~\ref{sec:method_extensions}. We develop the relevant theory in Section~\ref{sec:theory} and defer all the proofs to the supplementary material. 
Section~\ref{sec:simulations} examines the finite sample performance of our methods in terms of size, power and computational time  in comparison to some existing counterparts. Section~\ref{sec:conclusion} concludes. 
\\

\indent Now, we set up a few notations that will be used throughout this article. We denote the set of real numbers and integers by $\mathbb{R}$ and $\mathbb{Z}$, respectively. For any real number $x$, let us denote the greatest integer not exceeding $x$ by $\lfloor x \rfloor$ and the smallest integer greater than or equal to $x$ by $\lceil x \rceil$. For any random element $Z$, we use $Z\sim P_Z$ to denote the distribution of $Z$ as $P_Z$. We denote the space of all real-valued functions on $[a,b]$ which are right continuous with left limits (commonly known as "cadlag" functions) as $D[a,b]$. For any sequence of random variables, we use $\overset{D}{\to}$ and $\overset{\mathbb{P}}{\to}$ to denote convergence in distribution and convergence in probability, respectively. The symbol $\leadsto$ denotes the weak convergence of an associated stochastic process in an appropriate space. Finally, we use $\overset{D}{=}$ to denote equality in distribution.
\section{Methodology}
\label{sec:mathematical setups}
\subsection{Review of RKHS}\label{subsection:review_rkhs}
Consider a symmetric, continuous, and positive-definite kernel $K$ on $\mathbb{R}^p$. Let us denote the reproducing kernel Hilbert space (RKHS) generated by $K$ and the corresponding inner product and norm by $\mathcal{H}(K)$, $\langle .,.\rangle$ and $\left\|.\right\|$, respectively. For a given probability distribution $Q$ on $\mathbb{R}^p$, the mean embedding of $Q$ into the RKHS $\mathcal{H}(K)$ is given by an element $\mu_Q\in \mathcal{H}(K)$, where $\mu_Q :\mathbb{R}^p\mapsto \mathbb{R}$ satisfies the following:
\begin{equation}\label{eqn:mean_embedding}
    \langle f, \mu_Q\rangle=\bE_{X \sim Q} f(X),\;\;\forall f \in \mathcal{H}(K).
\end{equation}
Mean embedding of probability distributions into RKHS is a very well-studied topic and the literature is vast; we refer the reader to \cite{fukumizu2004dimensionality}, \cite{sriperumbudur2010hilbert} and \cite{berlinet2011reproducing} for more details. The following lemma (Theorem 1 of \cite{sriperumbudur2010hilbert}) provides us with some sufficient conditions for the existence of the mean embedding $\mu_Q$.
\begin{lemma}[\cite{JMLR:v13:gretton12a}]\label{lemma:existence_uniqueness_mean_embedding}
    If $K(.,.)$ is measurable and $\bE_{X \sim Q}\sqrt{K(X,X)}$ is finite, then mean embedding of $Q$ into $\mathcal{H}(K)$ exists.
\end{lemma}
For a certain class of kernels, defined as \say{characteristic} kernels, the mean embedding of the kernel of each probability distribution $Q$ that satisfies the conditions of Lemma~\ref{lemma:existence_uniqueness_mean_embedding} exists uniquely. \cite{sriperumbudur2010hilbert} provide the following clean characterization of such kernels.
\begin{lemma}[\cite{sriperumbudur2010hilbert}]
    Let $M$ be a topological space and let $K$ be an integrally strictly positive definite kernel on $M$, i.e.,
    \[\int_M \,\int_M K(x,y)d\mu(x)\,d\mu(y)>0,\]
    for all finite non-zero signed Borel measures $\mu$ defined on $M$. Then, $K$ is characteristic on the space of all Borel probability measures on $M$.
\end{lemma}
Examples of some characteristic kernels on the Euclidean space include the followings:
\begin{itemize}
    \item Gaussian kernel: \(K(x,y):=\exp{\left(-\sigma \left\|x-y\right\|_2^2\right)},\;\sigma>0;\)
    \item Laplacian kernel: \(K(x,y):=\exp{\left(-\sigma \left\|x-y\right\|_1\right)},\;\sigma>0;\)
    \item Exponential kernel: \(K(x,y):=\exp{(-\sigma\,x^{\top}y)}\);
    \item Inverse Multiquadratic kernel: \(K(x,y):=(\sigma^2+\left\|x-y\right\|_2^2)^{-c},\;\sigma,c>0.\)
\end{itemize}
For more details on characteristic kernels, see \cite{fukumizu2007kernel,sriperumbudur2008injective,sriperumbudur2010hilbert}. For the sake of simplicity, we assume throughout this article that any kernel we consider is a characteristic kernel and hence the RKHS mean embedding map is injective.  
Note that the reproducing property of the RKHS implies that if $\mu_Q$ exists, then it holds that
\begin{equation}\label{eqn:mean_embedding_reproduce}
    \langle f, \mu_Q\rangle=\bE_{X \sim Q} f(X)=\bE_{X \sim Q} \langle K(X,.),f\rangle,\;\;\forall f \in \mathcal{H}(K).
\end{equation}
This leads us to the concept of mean and covariance operator of a random element in a separable Hilbert space $H$.
\begin{definition}\label{defn:mean_cov_Hilbert_space}
    Let $X$ be a random element in a separable Hilbert space $H$, associated with an inner product $\langle .,. \rangle _{H}$. We say that $X$ has mean $\bE X$ if $\bE \langle X, h \rangle=\langle \bE X, h\rangle$ for every $h \in H.$ The covariance operator of $X$ (if it exists) is a function $S:H\mapsto H$, defined as
    \begin{equation}
        \langle S h_1, h_2 \rangle_H= \bE\left[\langle X-\bE X, h_1\rangle_H \,\cdot\,\langle X-\bE X, h_2\rangle_H\right],\;\;\forall\;h_1,h_2\in H.
    \end{equation}   
\end{definition}
It is well known that the separability of the underlying space implies the separability of a RKHS for a continuous kernel (see \cite{berlinet2011reproducing}). Thus, the separability of the $p$-dimensional Euclidean space implies that $\mathcal{H}(K)$ is also separable. Therefore, according to the definition of the mean of a $H$-valued random element and \eqref{eqn:mean_embedding_reproduce}, we can conclude that $\mu_Q$ is not only the mean embedding of $Q$ into the RKHS, but also the mean of the random element $K(X,.)$ where $X\sim Q$. This observation would be crucial in our theoretical analysis later. 

\subsection{Goodness-of-Fit Testing}\label{sec:method_GOF}
Suppose we observe a $p$-dimensional stationary time series $\mathcal{Y}:=\{Y_1,Y_2,\cdots,Y_n\}$ with marginal distribution $P$. For a given distribution $P_0$ on $\mathbb{R}^p$, we would like to test whether $P_0$ is the correct marginal distribution, i.e.,
$$H_0: P=P_0\text{   versus   }H_1:P\neq P_0.$$
In terms of mean embeddings, we can reformulate the problem as testing 
$$H_0:\mu_P-\mu_{P_0}=\mathbf{0}\in\mathcal{H}(K)\text{   versus  }H_1:\mu_P- \mu_{P_0}\neq \mathbf{0},$$
i.e.,
$$H_0:\left\|\mu_P-\mu_{P_0}\right\|^2=0\text{   versus  }H_1:\left\|\mu_P-\mu_{P_0}\right\|^2\neq 0.$$
The quantity $\left\|\mu_P-\mu_{P_0}\right\|$ is precisely the Maximum Mean Discrepancy (MMD) between the distributions $P$ and $P_0$, denoted by $\operatorname{MMD}(P,P_0).$  Use of MMD is very popular in the context of two-sample testing (see \cite{NIPS2006_e9fb2eda,gretton2009fast,JMLR:v13:gretton12a}, \cite{gaoshao23}) and it has also been used in the context of goodness-of-fit testing (see \cite{kellner2019one}). However, these works are developed for independent data only.\\

Now, we introduce our sample splitting approach and provide some intuition behind this approach. To begin with, consider the following stationary functional time series on the separable Hilbert space $\mathcal{H}(K)$: $\mathcal{X}:=\{X_1,X_2,\cdots,X_n\}$, where $X_i(.):=K(Y_i,.)-\mu_{P_0}(.)$, $i=1,2,\cdots,n.$ Note that our discussion in Section~\ref{subsection:review_rkhs} indicates that testing  $H_0$ is equivalent to testing whether the Hilbert space mean of $X_i$ is $\mathbf{0}$ or not. Now, we fix the splitting ratio $\eta\in (0,1)$ and let $m_1=\lfloor n \eta \rfloor$ and $m_2=n-m_1$. We split the sample $\mathcal{X}$ into two parts: $\mathcal{X}_1:=\{X_1,X_2,\cdots,X_{m_1}\}$ and $\mathcal{X}_2:=\{X_{m_1+1},\cdots,X_{n}\}$. Based on the first subsample $\mathcal{X}_1$, we estimate $\bE X_i=\mu_P-\mu_{P_0}$ by the function $\hat{\mu}_1$, where \
\begin{equation}\label{eqn:mu_1_hat_GOF}
   \hat{\mu}_1(.):=\frac{1}{m_1}\sum_{i=1}^{m_1} K(Y_i,.)-\mu_{P_0}(.), 
\end{equation}
 i.e. the sample mean over $\mathcal{X}_1$. Then, we project the data points in $\mathcal{X}_2$ along $\hat{\mu}_1$ using the RKHS inner product. This leads us to the following one-dimensional process: for $i=1,2,\cdots,m_2$, 
\begin{align*}
    Z_i:=\langle \hat{\mu}_1, X_{m_1+i}\rangle&=\frac{1}{m_1}\sum_{j=1}^{m_1}\langle K(Y_j,.)-\mu_{P_0},K(Y_{m_1+i},.)-\mu_{P_0}\rangle\\
    &=\frac{1}{m_1}\sum_{j=1}^{m_1} \left(K(Y_j,Y_{m_1+i})-\mu_{P_0}(Y_j)-\mu_{P_0}(Y_{m_1+i})+\left\|\mu_{P_0}\right\|^2\right),
\end{align*}
where the last equality follows from the reproducing property of RKHS, i.e., $\langle f, K(Y,.)\rangle=f(Y)$ for any element $f$ in the RKHS. When the original series is iid, the mean of $Z_i$ is exactly equal to $\operatorname{MMD}^2(P,P_0)$ (see Lemma 6 of \cite{JMLR:v13:gretton12a}). Under suitable weak dependence, the mean of $Z_i$ would be very close to $\operatorname{MMD}^2(P,P_0)$ for large enough $i$ (see Lemma 3.1. of \cite{dehling1983limit}). Therefore, we estimate $\operatorname{MMD}^2(P,P_0)$ by the sample mean of $Z_1,Z_2,\cdots Z_{m_2}$, i.e., $T_n:=\frac{1}{m_2}\sum_{i=1}^{m_2}Z_i.$ In order to normalize $T_n$, we have to take the complex temporal dependence of the $Z_i$s into account. To this end,  we use the self-normalization (SN) technique proposed by \cite{shao2010self}. In particular, we define our self-normalizer $W_n$ and the self-normalized test statistic $U_n$ by
\begin{equation}
\label{eqn:Un_GOF}
    W_n^2:=\frac{1}{m_2^2}\sum_{t=1}^{m_2}\left(\sum_{j=1}^t Z_j-t\, T_n\right)^2~\mbox{and}~ U_n:=\frac{T_n}{W_n},
\end{equation}
respectively. 
Under the null hypothesis of $\operatorname{MMD}^2(P,P_0)=0$, and assuming some appropriate moment and mixing conditions on the process $\{Y_t\}_{t \in \mathbb{Z}}$, we establish in Theorem~\ref{thm:asymp_null_GOF} that as $n \to \infty$,
\begin{equation}\label{eqn:defn_of_U}
    \sqrt{m_2}\cdot U_n\overset{D}{\to}U:=\frac{B(1)}{\sqrt{ \int_{0}^{1} \left[B(s)-sB(1)\right]^2 ds}},
\end{equation}
where $\{B(r):r\in [0,1]\}$ is a standard Brownian motion on $[0,1].$ Therefore, for a given significance level $\alpha$, we perform a one-sided test by rejecting $H_0$ if $U_n>U_{1-\alpha}/\sqrt{m_2}$, where $U_{1-\alpha}$ is the $(1-\alpha)$-th quantile of $U$. The distribution of $U$ is pivotal and and it does not depend on the sample-splitting ratio $\eta$. The quantiles of $U^2$ have been tabulated in \cite{lobato2001testing} via Monte Carlo simulations. This yields a resampling-free test.

\begin{remark}
    Since sample-splitting usually leads to a loss in efficiency, one may wish to use different sample-splitting ratios and use an aggregated test statistic (for example, a cross-fitted statistic). However, as \cite{kim2024dimension} have showed, the asymptotic null distribution of a cross-fitted statistic is non-pivotal even for iid data; see Proposition A.1 of Appendix A of the supplementary material therein.  
    As we will empirically demonstrate in Section~\ref{sec:simulations}, the sample splitting ratio has no significant impact on efficiency as long as it is chosen in a reasonable range.
\end{remark}
\begin{remark}
    Note that the implementation of our method requires the computation of the kernel mean embedding under the null, which is typical of the one-sample problems. Notably, \cite{kellner2019one} establish the explicit form of the mean embedding of a Gaussian distribution for Gaussian and exponential kernels. However, we emphasize that the  \noindent computation of a mean embedding in general is not straightforward. As we will see in the subsequent sections, our methodologies developed for change point and independence testing do not suffer from this limitation.
\end{remark}

\subsection{Change Point Testing}\label{sec:method_CPD}
Testing for structural stability in a time series is a problem of immense practical interest. In the presence of a change point, doing inference while assuming stationarity can lead to misleading results. Given a $p$-dimensional time series $\{Y_1,Y_2,\cdots,Y_n\}$, where the marginal distribution of $Y_i$ is $P_i$, we aim to test 
$$H_0:P_1=P_2=\cdots=P_n\text{   versus   } H_1:P_1=P_2=\cdots=P_{k_0}\neq P_{k_0+1}=P_{k_0+2}=\cdots=P_n,$$
where $k_0=\lfloor n \eta_0 \rfloor$ is the unknown change point location and $\eta_0\in (0,1).$ In terms of mean embeddings, it is equivalent to test
$$H_0:\mu_{P_1}=\mu_{P_2}=\cdots=\mu_{P_n}\text{   versus   } H_1:\mu_{P_1}=\mu_{P_2}=\cdots=\mu_{P_{k_0}}\neq \mu_{P_{k_0+1}}=\mu_{P_{k_0+2}}=\cdots=\mu_{P_n}.$$
For inference of a change point in marginal distribution in the context of temporally dependent data, we are only aware of \cite{inoue2001testing}, \cite{sharipov2016sequential} and \cite{bucchia2017change}. However, all the methods  in the above references are computationally quite expensive and numerical studies in those papers suggest sensitivity of the methods with respect to the choice of the bandwidth or block length involved in their resampling schemes. Here, we propose a novel sample-splitting and kernel based approach in this context which is quite fast and also has improved size accuracy and reduced sensitivity to tuning parameters. \\

\indent First, we transform the original data into the following functional time series on $\mathcal{H}(K)$: $\mathcal{S}:=\{F_i\}_{i=1}^n$, where $F_i(.)=K(Y_i,.)$. Note that the above change point testing problem could be regarded  as a two-sample mean testing problem in terms of $\mathcal{S}$ if the change point location was known apriori. 
To begin with, we fix the splitting ratio $\eta \in (0,1/2)$ and let $m_1=\lfloor n \eta \rfloor$ and $N=n-2m_1$. We split the sample  $\mathcal{S}$ into the following three parts: $\mathcal{S}_1:=\{F_1,F_2,\cdots F_{m_1}\}$, $\mathcal{S}_2:=\{F_{m_1+1},F_{m_1+2},\cdots,F_{n-m_1}\}$ and $\mathcal{S}_3:=\{F_{n-m_1+1}, F_{n-m_1+2},\cdots,F_n\}$. Based on the first and third subsample, we estimate $\mu_{P_1}-\mu_{P_n}$ by $\hat{\mu}_1-\hat{\mu}_n$, where
\begin{equation}\label{eqn:CPD_projection}
  \hat{\mu}_1(.)-\hat{\mu}_n(.):=\frac{1}{m_1}\sum_{i=1}^{m_1}F_i(.)-\frac{1}{m_1}\sum_{j=n-m_1+1}^{n}F_j(.)=\frac{1}{m_1}\sum_{i=1}^{m_1}K(Y_i,.)-\frac{1}{m_1}\sum_{j=n-m_1+1}^{n}K(Y_j,.).  
\end{equation}
We now project the data points in $\mathcal{S}_2$ along $\hat{\mu}_1-\hat{\mu}_n$ to form the following one-dimensional process:
\begin{align*}
    Z_k:=\langle \hat{\mu}_1-\hat{\mu}_n, F_{m_1+k}\rangle = \frac{1}{m_1}\sum_{i=1}^{m_1}K(Y_i,Y_{m_1+k})-\frac{1}{m_1}\sum_{j=n-m_1+1}^{n}K(Y_j,Y_{m_1+k}),
\end{align*}
where $k=1,2,\cdots,N.$ For independent data, the difference between the means of $Z_{k_0-m_1}$ and $Z_{k_0-m_1+1}$ would be exactly $\left\|\mu_{P_1}-\mu_{P_n}\right\|^2$. Under suitable weak dependence, we expect the amount of mean shift at location $k_0-m_1$ for $\{Z_k\}_{k=1}^{N}$ to be approximately the same.  Therefore, the original testing problem reduces to testing for a change in the mean of the univariate sequence $\{Z_k\}_{k=1}^{N}$, and we opt to use \cite{shao2010testing}'s SN-based method. To this end, define the cumulative sum $Z_{i:j}$  as $Z_{i:j}:=\sum_{k=i}^jZ_k$. For $k=1,2,\cdots,N$, we define the CUSUM statistic as
\begin{align*}
    T_n(k):=N^{-1/2}\sum_{t=1}^k \left(Z_t-\frac{1}{N}Z_{1:N}\right),
\end{align*}
and define the self-normalizer as
\begin{align*}
    V_n(k):=N^{-2}\left(\sum_{t=1}^k\left(Z_{1:t}-\frac{t}{k}Z_{1:k}\right)^2+\sum_{t=k+1}^{N}\left(Z_{t:N}-\frac{N-t+1}{N-k}Z_{k+1:N}\right)^2\right).
\end{align*}
Finally, our self-normalized test statistic is given by
\begin{equation}\label{eqn:Gn_CPD}
    G_n=\sup_{k=1,\cdots,N-1}T_n(k)/V^{1/2}_n(k).
\end{equation}
Under suitable  conditions on the process $\{Y_t\}_{t \in \mathbb{Z}}$, we establish in Theorem~\ref{thm:asymp_null_CPD} that under the null hypothesis of no change point, it holds that $G_n \overset{D}{\to}G,$ where
\begin{equation}\label{defn:Defn_of_G}
    G:=\sup_{r\in[0,1]}\;\frac{B(r)-rB(1)}{\left(\int_{0}^{r}(B(s)-\frac{s}{r}B(r))^2 ds+\int_{r}^{1}(B(1)-B(s)-\frac{1-s}{1-r}(B(1)-B(r)))^2 ds\right)^{1/2}}.
\end{equation}
Therefore, for a given significance level $\alpha$, we perform a one-sided test by rejecting $H_0$ if $G_n>G_{1-\alpha}$, where $G_{1-\alpha}$ is the $(1-\alpha)$-th quantile of $G$. The distribution of $G$ is pivotal and its quantiles have been tabulated in \cite{gao2023dimension} via Monte Carlo simulations.
\begin{example}
   In this example, we illustrate the application of our change point testing method in a real life dataset. We test for the presence of a change point in marginal distribution of the GNP (Gross National Product) dataset, as
analyzed in page 138 of \cite{shumway2006time}. The data
are U.S. quarterly U.S. GNP in billions of chained 1996 dollars from 1947(1) to 2002(3). The data has been seasonally
adjusted and we look at
the difference of the logarithm of the GNP, which is naturally
interpreted as the growth rate of GNP (see \cite{shumway2006time} for more details). The authors believe the growth rate to be stable and model it with various stationary time series models, such as AR(1) and MA(2). We apply our change point testing method on the difference of the logarithm of the GNP with sample splitting ratio $\eta=0.1$ and the Gaussian kernel with bandwidth selected via median heuristics,i.e.,
\[\sigma=Median\left\{\frac{1}{2|Y_{i_1}-Y_{i_2}|^2}, 1\leq i_1\neq i_2 \leq n\right\}.\]
The value of the SS-SN test statistic and its simulated p-value came out to be $8.98$ and $0.01$, respectively. This strongly indicates that there is a change point in the marginal distribution of the series, therefore  modelling the data with a stationary process is not appropriate. The location of the change point, estimated by the point of maximum of the SS-SN test statistic, came out to be $127$, which corresponds to the first quarter of 1979. This is consistent with the findings of \cite{shao2010testing}, who detected the presence of a change point for the 75$\%$ quantile.
\end{example}
\subsection{Testing Independence between Two Time Series }\label{sec:method_extensions}

As we have already hinted in our discussion in Section~\ref{sec:method_GOF}, the goodness-of-fit testing problem is equivalent to testing whether the mean of a set of stationary functional observations is the $\mathbf{0}$ element of an appropriate Hilbert space. Based on this observation, we extend our sample splitting methodology to develop a novel kernel-based test for testing the independence between two stationary processes. Checking independence between two time series is another inferential task of paramount importance. If the two series are dependent, one can further investigate the relationship between them via some causal analysis techniques, and that might lead the user to interesting discoveries or better predictive models. However, if they are indeed independent, one can proceed to model them independently using univariate time series models. Some notable works in this context include \cite{hong1996testing, hong2001testing, shao2009generalized} for two univariate time series and \cite{horvath2015testing} for two functional time series, among others.\\

\indent Since the seminal work of \cite{gretton2005measuring}, use of the Hilbert-Schmidt independence Criterion (HSIC), which is a kernel-based measure of independence capable of detecting both linear and non-linear dependence, has been very popular in statistics and machine learning. 
Some notable works for independent data include \cite{gretton2007kernel} and \cite{gretton2010consistent}. In the context of temporally dependent data, see \cite{zhang2008kernel}, \cite{zhou2012measuring}, \cite{wang2021new}, \cite{chu2023distance} etc. Note that \cite{zhang2008kernel} did not deal with inference. \cite{chu2023distance} and \cite{betken2021bootstrap} developed  distance covariance based methods to test the independence between two multivariate time series, whereas 
 the method in \cite{wang2021new} is based on HSIC. According to \cite{sejdinovic2013equivalence}, distance covariance can be viewed as a special case of HSIC. Below we point out a couple of drawbacks in these works. Firstly, the methodologies proposed by \cite{betken2021bootstrap} and \cite{wang2021new} suffer from large computational costs due to the use of block bootstrap. Moreover, the finite sample performance of the methods seem to be sensitive to the block length involved in the block bootstrap or the bandwidth involved in the kernel smoothing approach adopted by \cite{chu2023distance}. In contrast, we develop a fully nonparametric test based on HSIC that has very low computational cost because of the sample splitting and self-normalization approach and also has reduced sensitivity to tuning parameters in finite samples. From the theoretical angle, we only assume mild moment and strong mixing conditions on the data generating process to develop our asymptotic theory. \\

\indent Suppose we observe two stationary time series $\{X_t\}_{t=1}^{n}$ and $\{Y_t\}_{t=1}^{n}$ on $\mathbb{R}^p$ and $\mathbb{R}^s$, respectively. Our aim is to test the following hypothesis:
$$H_0: \{X_t\}_{t\in \mathbb{Z}}\indep \{Y_t\}_{t\in \mathbb{Z}}\text{ vs }H_1: H_0\text{ is not true }.$$
We fix a lag $m$, where $m$ is a non-negative integer such that $2m<n$. Under $H_0$, it holds that $H_0'(m)$: 
$X_0 \indep (Y_{-m},Y_{-m+1},\cdots,Y_0,\cdots, Y_{m-1}, Y_m)\text{ and } Y_0\indep (X_{-m},X_{-m+1},\cdots,X_0,\cdots, X_{m-1}, X_m).$
Now consider some continuous, positive-definite kernels $K_1, L_1, K_{2m+1}, L_{2m+1}$ on $\mathbb{R}^p$,  $\mathbb{R}^s$, $\mathbb{R}^{(2m+1)p}$ and $\mathbb{R}^{(2m+1)s}$, respectively.
In order to rigorously formulate the null hypothesis in terms of HSIC, we first consider two continuous, positive-definite kernels $K$ and $L$ on $\mathbb{R}^p$ and $ \mathbb{R}^{mp}$, respectively. We define $C_{X_0,Y_{-m:m}}$ as the associated cross-covariance operator between $X_0$ and $(Y_{-m},\cdots,Y_{m})$, i.e.,
\begin{equation}\label{eqn:cross_cov_indep_upo_lag_m}
C_{X_0,Y_{-m:m}}=\bE \left[\left(K_1(X_0,.)-\mu_1\right)\otimes \left(L_{2m+1}((Y_{-m},\cdots,Y_{m}),.)-\nu_2\right)\right],
\end{equation}
where $\mu_1:=\bE K_1(X_0,.)$, $\nu_2:=\bE L_{2m+1}((Y_{-m},\cdots,Y_{m}),.)$, and $\otimes $ denotes the tensor product of two Hilbert-valued elements. We can similarly define $C_{Y_0, X_{-m:m}}$, i.e.,
\begin{equation}\label{eqn:cross_cov_indep_upo_lag_m_2}
C_{Y_0,X_{-m:m}}=\bE \left[\left(L_1(Y_0,.)-\nu_1\right)\otimes \left(K_{2m+1}((X_{-m},\cdots,X_{m}),.)-\mu_2\right)\right],
\end{equation}
where $\nu_1:=\bE L_1(Y_0,.)$, $\mu_2:=\bE K_{2m+1}((X_{-m},\cdots,X_{m}),.)$
For the basic details on the tensor product of two Hilbert spaces and the HSIC, we refer the reader to Appendix \ref{appnb}.  Now, in light of Lemma B.1. in the supplementary material, the weak null hypothesis of our interest $H_0'(m)$ is equivalent to
$H_0'(m):\begin{pmatrix}
    C_{X_0,Y_{-m:m}}\\
    C_{Y_0,X_{-m:m}}
\end{pmatrix}=\mathbf{0}$, 
where the associated Hilbert space is $\mathcal{H}(K,L;m):=\begin{pmatrix}
    \mathcal{H}(K_1)\otimes \mathcal{H}(L_{2m+1})\\
    \mathcal{H}(L_1)\otimes \mathcal{H}(K_{2m+1})
\end{pmatrix}$, equipped with the inner product
$$\left\langle \begin{pmatrix}
    f_1\otimes g_1\\
    f_2 \otimes g_2
\end{pmatrix},\,\begin{pmatrix}
    f_3\otimes g_3\\
    f_4 \otimes g_4
\end{pmatrix}\right\rangle=\langle f_1,f_3\rangle\cdot \langle g_1, g_3\rangle+\langle f_2,f_4\rangle\cdot \langle g_2, g_4\rangle.$$

\noindent It is easy to check that this Hilbert space is also separable. Therefore, the null hypothesis is equivalent to $\left\|C_{X_0,Y_{-m:m}}\right\|_{HS}= \left\|C_{Y_0,X_{-m:m}}\right\|_{HS}=0$. Once again, we emphasize the similarity of the above testing framework to the one in Section~\ref{sec:method_GOF}. In both scenarios, we are essentially interested in testing whether the mean of a set of temporally dependent functional sequence is the zero element of an appropriate Hilbert space or not. However, in contrast to the goodness-of-fit setting, we do not need to assume the marginal distributions to be known under $H_0$. As we will see in this section and also in Section~\ref{sec:theory_indep_twoTS}, it suffices to estimate $\mu_1,\mu_2,\nu_1\text{ and }\nu_2$ with their corresponding sample means and the test does not involve explicit computation of the mean embeddings which can be pretty difficult in practice. We now describe our  methodology.
Let $n'=n-2m$ and for $k=1,2,\cdots n'$, we define
\begin{equation}\label{eqn:Ei_twoTS}
    E_k=\begin{pmatrix}
    (K_1(X_{k+m},.)-\hat{\mu}_1)\otimes (L_{2m+1}((Y_{k},Y_{k+1},\cdots,Y_{k+m},\cdots, Y_{k+2m}),.)-\hat{\nu}_{2})\\
    (L_1(Y_{k+m},.)-\hat{\nu}_1)\otimes (K_{2m+1}((X_{k},X_{k+1},\cdots,X_{k+m},\cdots, X_{k+2m}),.)-\hat{\mu}_{2})
\end{pmatrix},
\end{equation}
where 
$\hat{\mu}_1=\frac{1}{n'}\sum_{i=m+1}^{n-m}K_1(X_i,.),\;\;\hat{\nu}_1=\frac{1}{n'}\sum_{i=m+1}^{n-m}L_1(Y_i,.),$
$$\hat{\mu}_2=\frac{1}{n'}\sum_{i=m+1}^{n-m}K_{2m+1}((X_{i-m},X_{i-m+1},\cdots,X_{i},\cdots, X_{i+m-1}, X_{i+m}),.),$$
and
$$\hat{\nu}_2=\frac{1}{n'}\sum_{i=m+1}^{n-m}L_{2m+1}((Y_{i-m},Y_{i-m+1},\cdots,Y_{i},\cdots, Y_{i+m-1}, Y_{i+m}),.).$$  
\noindent  Next let $m_1'=\lfloor n' \eta\rfloor$ and $m_2':=n'-m_1'.$ We form a direction of projection as $\hat{\lambda}:=\frac{1}{m'_1}\sum_{i=1}^{m'_1}E_i$ and then project the points $\{E_k\}_{k=m'_1+1}^{n'}$ along $\hat{\lambda}$ to form the following univariate sequence: $F_j^{(1)}:=\left\langle \hat{\lambda}, E_{j+m_1'}\right\rangle,$ where $j=1,2,\cdots,m'_2$. Now, we let $T^{(1)}_n:=\frac{1}{m_2'}\sum_{k=1}^{m_2'}F^{(1)}_k$ and use the self-normalizer ${W_n^{(1)}}=\sqrt{\frac{1}{m_2'^2}\sum_{t=1}^{m_2'}\left(\sum_{j=1}^t F^{(1)}_j-t\, T^{(1)}_n\right)^2}$. Finally, we define the test statistic given by \begin{equation}\label{eqn:test_stat_indep_twoTS}
    U^{(1)}_n:=\frac{T^{(1)}_n}{W^{(1)}_n}.
\end{equation}
In Theorem~\ref{thm:asymp_indep_lag_m}, we show that under some mild regularity conditions, the asymptotic null distribution of the test statistic $U^{(1)}_n$ is given by $U$, where $U$ is pivotal as defined in~\eqref{eqn:defn_of_U}. This facilitates us to calibrate our test statistic for a one-sided test via the quantiles of $U.$
\begin{remark}
    It is worthwhile to note that unlike \cite{betken2021bootstrap}, our proposed test targets only a particular sublcass of the original null hypothesis $H_0$. However, we would like to mention that it is quite straightforward to extend our methodology to a null hypothesis such as
    $H''_0(m) : (X_{-s},X_{-s+1}, \cdots ,X_0, \cdots ,X_s) \indep (Y_{-m}, Y_{-m+1}, \cdots , Y_0, \cdots , Y_m)$. In particular, one can just work with 
    $$\{(K_{2s+1}((X_{k},X_{k+1}, \cdots , \cdots ,X_{k+2s}),\cdot)-\hat{\mu}_s)\otimes (L_{2m+1}((Y_{k}, Y_{k+1}, \cdots , Y_{k+2m}),\cdot)-\hat{\nu}_m)\}$$ and define the test statistic similarly, where $\hat{\mu}_s$ and $\hat{\nu}_m$ are corresponding sample averages. Moreover, in a similar fashion, our methodology can be extended to the framework of serial independence testing of a time series up to a fixed lag. Fore the sake of conciseness, we do not get into more details of such extension in the paper.
\end{remark}
\section{Asymptotic Theory}\label{sec:theory}
The key technical assumption in our theoretical analysis is a functional central limit theorem (FCLT) in the associated RKHS. We refer the reader to Appendix \ref{appna} for the definitions of a Gaussian random element, a Brownian motion in a Hilbert space and the space of $\mathcal{H}(K)$-valued cadlag (right continuous function which has left limit at all points) functions on $[0,1]$, denoted by $D_{\mathcal{H}(K)}[0,1]$. Let us recall that in the goodness-of-fit and change point testing problems, $\{Y_1,\cdots,Y_n\}$ comes from a stationary process.
\begin{assumption}\label{ass:FCLT_GOF_NULL}
    It holds under $H_0$ that
    \begin{equation}\label{eqn:FCLT_RKHS}
    \left(\frac{1}{\sqrt{n}} \sum_{i=1}^{\lfloor n t\rfloor}\left(K(Y_i,.)-\mu_P\right)\right)_{t \in[0,1]} \leadsto (W(t))_{t \in[0,1]}\text{ in } D_{\mathcal{H}(K)}[0,1],
\end{equation}
where $P$ is the marginal distribution of $Y_1$, $(W(t))_{t \in[0,1]}$ is a Brownian motion in $\mathcal{H}(K)$ and $W(1)$ is a zero-mean random element with covariance operator $S: \mathcal{H}(K) \rightarrow \mathcal{H}(K)$, defined by
\begin{equation}\label{eqn:covar_FCLT_RKHS}
    \langle S x, y\rangle=\sum_{i=-\infty}^{\infty} \bE\left[\left\langle K(Y_0,.)-\mu_P, x\right\rangle\left\langle K(Y_i,.)-\mu_P, y\right\rangle\right], \text { for } x, y \in \mathcal{H}(K) .
\end{equation}
Furthermore, the series in \eqref{eqn:covar_FCLT_RKHS} converges absolutely.
\end{assumption}
\begin{remark}
    Note that in the goodness-of-fit problem, $P=P_0$ is pre-specified under the null, whereas $P$ is unspecified in the change point problem.
\end{remark}

FCLT for Hilbert-valued temporally dependent sequence is a well-studied topic in the literature, see \cite{chen1998central, sharipov2016sequential}. In particular, \cite{chen1998central} study the FCLT for a Hilbert-valued sequence $\mathcal{L}_2$-near epoch dependent (NED) on a strong mixing process. But, \cite{sharipov2016sequential} assume absolute regularity, which is a stronger assumption than strong mixing. However, they assume $\mathcal{L}_1$-NED, which is more general than $\mathcal{L}_2$-NED. Fortunately, we show in the following proposition that the conditions of \cite{chen1998central} are satisfied by the RKHS-embedded series if the underlying data-generating process is stationary, strong mixing at a certain rate and satisfies a mild moment condition. We refer the reader to Appendix \ref{appnc} and \ref{appnd} for the technical details and proofs.
\begin{proposition}\label{prop:FCLT_Verification_GOF_NULL}
    Suppose $\{Y_i\}_{i \in \mathbb{Z}}$ forms a stationary, strongly mixing process on $\mathbb{R}^p$ with marginal distribution $P$ and $\{\alpha(m)\}_{m \in \mathbb{N}}$ are its mixing coefficients (see Appendix \ref{appnc} for the definition). Moreover, suppose the followings hold for some $\delta>0$:
    \begin{itemize}
        \item [(i)]\(\bE\, |K(Y_1,Y_1)|^{1+\delta/2}<\infty;\)
        \item [(ii)] \(\sum_{m=1}^{\infty} (\alpha(m))^{\delta /(2+\delta)}<\infty.\)
    \end{itemize}
    Then, Assumption~\ref{ass:FCLT_GOF_NULL} is satisfied.
\end{proposition}
\begin{remark}
    Any bounded kernel, which includes Gaussian, Laplacian, Exponential, Student, etc., satisfies $(i)$ in Proposition~\ref{prop:FCLT_Verification_GOF_NULL} and thus, is valid on heavy-tailed data as well. In Section~\ref{sec:simulations_FDA}, we demonstrate the excellent size accuracy of our proposed test for change point in heavy-tailed data. Part $(ii)$ is a standard assumption in the literature of central limit theorem for $\beta$-mixing processes; see \cite{ibragimov1962some, dehling1983limit, doukhan1994mixing}. It requires the underlying temporal dependence to be of short-range type and is satisfied if $\alpha(m)$ decays at least at a polynomial rate of $m^{-(1+\epsilon)(1+2/\delta)}$ for some $\epsilon>0$. It has been established in the literature that stationary Markov chains that possess geometric ergodicity are strong mixing at an exponential rate; see \cite{bradley2005basic, meyn2012markov} for more details. \cite{mokkadem1988mixing} shows that a stationary vector ARMA process is geometrically strong mixing if the iid innovations have a density with respect to the Lebesgue measure. \cite{carrasco2002mixing} provide sufficient conditions for geometric $\beta$-mixing (which implies geometric strong mixing) for various linear and nonlinear GARCH$(1,1)$, linear and power GARCH$(p,q)$,
 stochastic volatility, and autoregressive conditional duration models. For more details, see the aforementioned references. 
\end{remark}
\begin{remark}
    Our RKHS-embedding approach has a slight advantage than the one in \cite{Zhang27032025} on the theory front as well. The asymptotic theory behind their proposed tests for the marginal distribution  relies on the uniform FCLT for the canonical empirical process. The uniform FCLT 
 has been proved under $\rho$-mixing by \cite{berkes1977almost}, and by \cite{berkes2009asymptotic} under S-mixing; for definitions see therein. It is well-known that $\rho$-mixing is a stronger assumption than strong mixing. While S-mixing is easier to verify for certain models, it is certainly restricted to a limited class of models which admit a particular representation. Therefore, the strong mixing assumption in Proposition~\ref{prop:FCLT_Verification_GOF_NULL} is more general in our opinion.
\end{remark}
We are now ready to develop the asymptotic theory for the tests developed in the previous section.
\subsection{Goodness-of-Fit Testing}
For this subsection, let $U_n$ and $U$ be defined as in \eqref{eqn:Un_GOF} and \eqref{eqn:defn_of_U}, respectively.
\begin{theorem}\label{thm:asymp_null_GOF}
    Suppose Assumption~\ref{ass:FCLT_GOF_NULL} holds with $P=P_0$. Then under $H_0$, as $n \to \infty$,
    $$\sqrt{m_2}\cdot U_n\overset{D}{\to}U.$$
\end{theorem}
Next, we investigate the limiting power of our goodness-of-fit test under the alternative hypothesis that there is a distributional shift from the pre-specified distribution. We consider the following sequence of local alternatives:
$$H_{1n}:P=P_n,\text{ where }\mu_{P_n}\in \mathcal{H}(K),\;\;\mu_{P_n}-\mu_{P_0}=\delta_n\text{ and }\left\|\delta_n\right\|_{\mathcal{H}(K)}=\frac{\Delta_n}{\sqrt{n}}.$$
We make the following assumption on the data-generating process under the alternative hypothesis.
\begin{assumption}\label{ass:FCLT_local_alternative_GOF}
   Under $\left\{H_{1n}\right\}_{n=1}^{\infty}$, it holds that
    \begin{equation}
    \left(\frac{1}{\sqrt{n}} \sum_{i=1}^{\lfloor n t\rfloor}\left(K(Y_i,.)-\mu_{P_n}\right)\right)_{t \in[0,1]} \leadsto (\tilde{W}(t))_{t \in[0,1]}\text{ in } D_{\mathcal{H}(K)}[0,1],
\end{equation}
where $(\tilde{W}(t))_{t \in[0,1]}$ is a Brownian motion in $\mathcal{H}(K)$ and $\tilde{W}(1)$ is a zero-mean random element with covariance operator $\tilde{S}: \mathcal{H}(K) \rightarrow \mathcal{H}(K)$.
\end{assumption}
Similar to Proposition~\ref{prop:FCLT_Verification_GOF_NULL}, we show in the following proposition that when $(2+\delta)$-th moments are uniformly bounded and the strong mixing occurs at a suitable rate, a couple of additional assumptions on uniform integrability and convergence of variance implies Assumption~\ref{ass:FCLT_local_alternative_GOF} is satisfied.
\begin{proposition}\label{prop:verify_FCLT_GOF_Alt}
    Suppose $\{Y_i\}_{1\leq i \leq n}$ comes from a triangular array of random variables which is row-wise stationary and strong mixing and marginal distribution of random variables in the $n$-th row is $P_n$. Let the mixing coefficients at lag $m$ be uniformly bounded by $\tilde{\alpha}(m)$. Moreover, assume that followings hold for some $\delta>0$:
    \begin{itemize}
        \item [(i)] \(\bE |K(Y_i,Y_i)|^{1+\delta/2},\,\left\|\mu_{P_n}\right\|\leq C<\infty,\) where $C>0$ does not depend on $n$ or $i$;
        \item [(ii)] \(\sum_{m=1}^{\infty}\tilde{\alpha}(m)^{\frac{\delta}{\delta+2}}<\infty\)
        \item [(iii)] \(\left\{\frac{\left\|K(\tilde{Y}_n,\cdot)-\mu_{P_n}\right\|^2}{\bE K(\tilde{Y}_n,\tilde{Y}_n)-\left\|\mu_{P_n}\right\|^2};n\geq 1\right\}\) is uniformly integrable where marginal distribution of the random variable $\tilde{Y}_n$ is $P_n$;
        \item [(iv)] For any $h\neq 0 \in \cH(K)$, as $n\to \infty$, the following holds:
        \[\frac{1}{n}Var \left(\sum_{1\leq i \leq \lfloor nt \rfloor} h(Y_i)\right)\to t\,\tilde{\sigma}^2(h)>0\text{ for all }0<t\leq 1. \]
    \end{itemize}
    Then, Assumption~\ref{ass:FCLT_local_alternative_GOF} is satisfied with $\langle \tilde{S}h,h\rangle=\tilde{\sigma}^2(h)$ for all $h\neq 0\in \cH(K)$.
\end{proposition}

\begin{theorem}\label{thm:GOF_asymptotic_power}
  Suppose Assumption~\ref{ass:FCLT_local_alternative_GOF} holds and let $b(\eta):=\frac{1}{\eta} \tilde{W}(\eta).$ Then, under $\{H_{1n}\}_{n=1}^{\infty}$, \(\sqrt{n}\hat{\mu}_1\overset{D}{\to}b(\eta)\) as $n \to \infty$, where $\hat{\mu}_1$ is defined in \eqref{eqn:mu_1_hat_GOF}. Moreover, the followings hold:
    \begin{enumerate}
        \item If $\Delta_n^2\to 0$, then $\mathbb{P}(\sqrt{m_2}\cdot U_n>U_{1-\alpha}) \to \alpha$.
        \item If $\Delta_n^2 \to \infty$, then $\mathbb{P}(\sqrt{m_2}\cdot U_n>U_{1-\alpha}) \to 1$. 
        \item If $\Delta_n^2 \to c^2$, where $c \in (0,\infty)$, then whenever $\frac{\delta_n}{\left\|\delta_n\right\|}\to \delta_0 \in \mathcal{H}(K),$ it holds that 
        \begin{itemize}
            \item [] 
            \begin{align*}
          &\mathbb{P}(\sqrt{m_2}\cdot U_n>U_{1-\alpha})\\
          &\to \int_{\mathcal{H}(K)} \mathbb{P}\left(\frac{B(1)+\frac{c (1-\eta) \langle \delta_0, u \rangle}{\Theta(u)}+\frac{c^2 \eta (1-\eta)}{\Theta(u)}}{\sqrt{\int_{0}^{1}\left(B(s)-sB(1)\right)^2 ds}}>U_{1-\alpha}\right)\mu_{\tilde{W}(\eta)}(du),
        \end{align*}
        \end{itemize}
        where $\Theta(u):=\sqrt{\langle S(u+c \eta \delta_0), u+c \eta \delta_0\rangle}$, $\mu_{\tilde{W}(\eta)}$ is the distribution of $\tilde{W}(\eta)$ and the above integral is with respect to the measure $\mu_{\tilde{W}(\eta)}$ on the Hilbert space $\mathcal{H}(K)$.
    \end{enumerate}
\end{theorem}
\begin{remark}
    One can think of $\Delta_n$ as rescaled signal-to-noise ratio (SNR). Theorem~\ref{thm:GOF_asymptotic_power} reveals that when the SNR approaches $0$, our level-$\alpha$ test is powerless; when the SNR approaches $\infty$, our test is consistent; and when the SNR approaches an intermediate value between $0$ and $\infty$, we achieve non-trivial power.
\end{remark}
\begin{remark}
    Note that the assumption of one-to-one correspondence between the space of probability distributions and the space of mean embeddings is standard and quite crucial, otherwise MMD is not be a metric. It can be seen in the proof of the above results that size accuracy does not require this condition. However, under the violation of this assumption, if MMD$(P_n,P_0)=0$, then one will fail to distinguish between the null and the alternative and that will lead to serious power loss of our method. 
\end{remark}
\subsection{Change Point Testing}\label{sec:theory_CPD}
For this section, let $G_n$ and $G$ be defined as in \eqref{eqn:Gn_CPD} and \eqref{defn:Defn_of_G}, respectively. 
\begin{theorem}\label{thm:asymp_null_CPD}
        Suppose Assumption~\ref{ass:FCLT_GOF_NULL} holds. Then, under $H_0$, as $n \to \infty$, it holds that $G_n \overset{D}{\to}G.$
    \end{theorem}
  Next, we investigate the power of our test against the alternative hypothesis that there exists a single change point at the location $k_0=\lfloor n \eta_0 \rfloor$ for some fixed but unknown $\eta_0\in (\eta,1-\eta).$ Let us recall that in this setup, we denote the marginal distribution of $Y_i$ by $P_i$. In particular, we analyze the power behavior of our test against the following sequence of local alternatives:
    $$H'_{1n}: \mu_{P_1}=\cdots=\mu_{P_{k_0}}\neq\mu_{P_{k_0+1}}=\cdots=\mu_{P_n},\text{ where } \mu_{P_{k_0}+1}-\mu_{P_{k_0}}=\delta_n,\;\left\|\delta_n\right\|=\frac{\Delta_n}{\sqrt{n}} \in \mathbb{R}.$$
     We assume that the following FCLT holds under the alternative.
    \begin{assumption}\label{ass:FCLT_local_alternative_CPD}
        Under $\left\{H'_{1n}\right\}_{n=1}^{\infty}$, it holds that
    \begin{equation}\label{eqn:FCLT_RKHS_CP}
    \left(\frac{1}{\sqrt{n}} \sum_{i=1}^{\lfloor n t\rfloor}\left(K(Y_i,.)-\mu_{P_i}\right)\right)_{t \in[0,1]} \leadsto (\Tilde{W}(t))_{t \in[0,1]}\text{ in } D_{\mathcal{H}(K)}[0,1],
\end{equation}
where $(\Tilde{W}(t))_{t \in[0,1]}$ is a Brownian motion in $\mathcal{H}(K)$ and $\Tilde{W}(1)$ has covariance operator $\Tilde{S}: \mathcal{H}(K) \rightarrow \mathcal{H}(K)$.
    \end{assumption}
    \begin{remark}
        Very similar to Proposition~\ref{prop:verify_FCLT_GOF_Alt}, one can provide sufficient moment and mixing conditions, along with some uniform integrability conditions, which will imply that the FCLT under the alternative hypothesis holds. We omit the details for brevity.
    \end{remark}
Before we describe the asymptotic behavior of our change point test statistic under alternative, let us define 
\[G^{*}_{cond}:=\sup_{r \in [0,1]}\frac{ B(r)-rB(1)-\sqrt{\frac{1-2\eta}{\langle Su, u \rangle}}(r_0\wedge r)((1-r_0)\wedge (1-r))\langle u, c\delta_0 \rangle }{\sqrt{\splitdfrac{\int_{0}^{r}\left((B(s)-\frac{s}{r}B(r))-\sqrt{\frac{1-2\eta}{\langle Su, u \rangle}}\frac{[(s \wedge r_0)((r-r_0)\wedge (r-s))]\vee 0}{r}\langle u, c\delta_0 \rangle\right)^2 ds}%
              {\splitdfrac{+\int_{r}^{1}\Big((B(1-s)-\frac{1-s}{1-r}B(1-r))}{-\sqrt{\frac{1-2\eta}{\langle Su, u \rangle}}\frac{[((1-s) \wedge (1-r_0))((r_0-r)\wedge (s-r))]\vee 0}{r}\langle u, c\delta_0 \rangle\Big)^2 ds}}
            }}.\]
\begin{theorem}\label{thm:CPD_asymptotic_power}
    Suppose Assumption~\ref{ass:FCLT_local_alternative_CPD} holds and let $\tilde{b}(\eta):=\frac{1}{\eta}( \tilde{W}(\eta)-\tilde{W}(1)+\tilde{W}(1-\eta)).$ Then, under $\{H'_{1n}\}_{n=1}^{\infty}$,  $\sqrt{n}(\hat{\mu}_1-\hat{\mu}_n)\overset{D}{\to} \tilde{b}(\eta)$ as $n\to \infty$, where  $\hat{\mu}_1-\hat{\mu}_n$ is defined in \eqref{eqn:CPD_projection}. Moreover, the followings hold:
    \begin{enumerate}
        \item If $\Delta_n\to 0$, then $\mathbb{P}(G_n>G_{1-\alpha})\to \alpha.$
        \item If $\Delta_n\to \infty$, then $\mathbb{P}(G_n>G_{1-\alpha})\to 1.$
        \item If $\Delta_n\to c\in(0,\infty)$ and $\frac{\delta_n}{\left\|\delta_n\right\|}\to \delta_0\in \mathcal{H}(K)$, then $G_n \overset{D}{\to} G^*$, where
        \[G^*\mid \{\tilde{b}(\eta)=u+c\delta_0\}\overset{D}{=}G^*_{cond}.\]
    \end{enumerate}
\end{theorem}
\subsection{Testing Independence between Two Time Series}\label{sec:theory_indep_twoTS}
To derive the asymptotic distribution of our test statistic for the independence of two time series under the null and alternative hypothesis, we first assume that the following invariance principle holds.
\begin{assumption}\label{ass:FCLT_cross_cov_indep_upto_lag_m}
Let $C_{X_0,Y_{-m:m}}$, $C_{Y_0,X_{-m:m}}$ and $E_i$ be as defined in \eqref{eqn:cross_cov_indep_upo_lag_m}, \eqref{eqn:cross_cov_indep_upo_lag_m_2} and \eqref{eqn:Ei_twoTS}, respectively. Then, as $n\to \infty$, it holds that
  $$\left\{\frac{1}{\sqrt{n'}}\sum_{i=1}^{\lfloor n'r\rfloor}\left(E_i-\begin{pmatrix}
    C_{X_0,Y_{-m:m}}\\
    C_{Y_0,X_{-m:m}}
\end{pmatrix}\right)\right\}_{r \in [0,1]}\leadsto W \text{ in } D_{\mathcal{H}(K,L;m)}[0,1],$$
        where $W$ is a $\mathcal{H}(K,L;m)$-valued Brownian motion.
\end{assumption}
In the following proposition, similar to Proposition~\ref{prop:FCLT_Verification_GOF_NULL}, we show that the above invariance principle holds under the null hypothesis under mild moment and mixing conditions.
\begin{proposition}\label{prop:FCLT_cross_cov_indep_upto_lag_m_Verify}
    Let $\{(X_i,\,Y_{i})\}_{i\in \mathbb{Z}}$ be a stationary and strong mixing process with mixing coefficients as $\{\alpha(k)\}_{k\in \mathbb{Z}}$. Moreover, assume that the followings hold for some $\delta>0$:
\begin{enumerate}
    \item $\bE |K_1(X_0,X_0)|^{2+\delta},\;\bE |L_{2m+1}(Y_{-m:m},Y_{-m:m})|^{2+\delta},\; \bE |L_1(Y_0,Y_0)|^{2+\delta} \text{ and }$\\
    $\bE |K_{2m+1}(X_{-m:m},X_{-m:m})|^{2+\delta} $ are all finite.
    \item$\sum_{m=1}^{\infty} \alpha(m)^{\delta /(\delta+2)}<\infty$.
\end{enumerate}
Then, Assumption~\ref{ass:FCLT_cross_cov_indep_upto_lag_m} is satisfied.
\end{proposition}
\noindent The proof of Proposition~\ref{prop:FCLT_cross_cov_indep_upto_lag_m_Verify} can be found in Appendix \ref{appnc}. Assumption~\ref{ass:FCLT_cross_cov_indep_upto_lag_m} helps us to do inference without any further knowledge of the joint distributions. We are now ready to characterize the asymptotic behavior of our test statistic $U^{(1)}_n$ defined in \eqref{eqn:test_stat_indep_twoTS}.
\begin{theorem}\label{thm:asymp_indep_lag_m}
        Suppose Assumption~\ref{ass:FCLT_cross_cov_indep_upto_lag_m} holds. Then,
        \begin{enumerate}
            \item Under $H_0$, as $n \to \infty,$ it holds that
            $\sqrt{m_2'}\cdot U^{(1)}_n\overset{D}{\to}U,
$
where $U$ is defined in (\ref{eqn:defn_of_U}). 
            \item Under the local alternative $H_{1n}'':\begin{pmatrix}
    C_{X_0,Y_{-m:m}}\\
    C_{Y_0,X_{-m:m}}
\end{pmatrix}=\frac{\tilde{\delta}_n}{\sqrt{n'}}\in \mathcal{H}(K,L;m),$ as $n \to \infty$, the followings hold:
    \begin{enumerate}
        \item [(a)] If $\left\|\tilde{\delta}_n\right\|^2\to 0$, then $\mathbb{P}(\sqrt{m_2'}\cdot U^{(1)}_n>U_{1-\alpha}) \to \alpha$. 
        \item[(b)] If $\left\|\tilde{\delta}_n\right\|^2 \to \infty$, then $\mathbb{P}(\sqrt{m_2'}\cdot U^{(1)}_n>U_{1-\alpha}) \to 1$. 
        \item [(c)] If $\left\|\tilde{\delta}_n\right\|^2 \to c^2$, where $c \in (0,\infty)$, then whenever $\frac{\tilde{\delta}_n}{\left\|\tilde{\delta}_n\right\|}\to \tilde{\delta}_0 \in \mathcal{H}(K,L;m),$ we have that
       \begin{align*}
          &\mathbb{P}(\sqrt{m_2'}\cdot U^{(1)}_n>U_{1-\alpha})\\
          &\to \int \mathbb{P}\left(\frac{B(1)+\frac{c (1-\eta) \langle \tilde{\delta}_0, u \rangle}{\tilde{\Theta}(u)}+\frac{c^2 \eta (1-\eta)}{\tilde{\Theta}(u)}}{\sqrt{\int_{0}^{1}\left(B(s)-sB(1)\right)^2 ds}}>U_{1-\alpha}\right)\mu_{W(\eta)}(du),
        \end{align*}
         where $W$ is a Brownian motion in $\mathcal{H}(K,L;m)$,\\ $\tilde{\Theta}(u):=\sqrt{\langle S(u+c \eta \tilde{\delta}_0), u+c \eta \tilde{\delta_0}\rangle}$, and $S$ is the covariance operator of $W(1)$. Here $\mu_{W(\eta)}$ is the distribution of $W(\eta)$ and the above integral is a Lebesgue integral on the Hilbert space $\mathcal{H}(K,L;m)$.
    \end{enumerate}
        \end{enumerate}
    \end{theorem}
    \begin{remark}
        Testing for the presence of a change point, and the  independence between two series are  important tasks for functional time series as well; see \cite{gabrys2007portmanteau}, \cite{zhangshao2011},  \cite{horvath2012inference}, \cite{horvath2015testing} for the related literature. In these works,   the functional observations are modeled as random curves observed in $\mathcal{L}_2[0,1]$, which is the space of square-integrable real-valued functions on $[0,1]$.  Since our methods are constructed using kernels, they can be implemented in such setups. Moreover, as we pointed out 
 in Section~\ref{subsection:review_rkhs} and Section~\ref{sec:theory}, the theoretical guarantees still remain valid as long as the underlying space is a separable Hilbert space, which is indeed true for any continuous kernel on $\mathcal{L}_2[0,1]$, and certain weak dependence and moment  conditions. Again, note that the kernels need to be characteristic in $\mathcal{L}_2[0,1]$, else we may lose power against certain alternatives; see \cite{duncantwosample} for characterization of such kernels defined on a general Hilbert space, which includes the Gaussian kernel, the inverse multiquadratic kernel etc.  
    \end{remark}
    \begin{remark}
        Non-Euclidean data or object-valued data are becoming increasingly common in many fields, such as biological or social sciences. Examples of such data sets include graphs, networks, distributional data, covariance matrices, etc. Change point detection or testing for serial independence are crucial problems in such data sets as well; see \cite{dubey2020frechet}, \cite{dubey2023change}, \cite{jiang2024two}, \cite{jiang2024testing} for some of the recent developments on the above problems. Due to the use of kernels, our proposed methodology can be applied to object-valued temporally dependent data. Moreover, the asymptotic theory we have developed translates to such setup under some standard regularity conditions. Indeed, for a given metric space $(\Omega,d)$, \cite{zhang2024dimension} showed that the Gaussian and Laplacian kernels, defined as $K_1(x,y):=\exp{(-\sigma\,d^2(x,y))}$ and $K_2(x,y):=\exp{(-\sigma\,d(x,y))}$, are positive-definite. Moreover, they are universal kernels on such metric spaces, which implies that the mean embedding of any probability distribution on $\Omega$ into the RKHS generated by such a kernel, if it exists, is unique. We refer the reader to \cite{bhattacharjee2023nonlinear} for more details on this matter. 
    \end{remark}     
    \section{Numerical Experiments}\label{sec:simulations}
    In this section, we analyze the finite sample performance of our proposed methods in terms of empirical size, power, and computational time and compare them to those of some existing counterparts. Simulation results for goodness-of-fit testing are reported in Section~\ref{sec:simulations_GOF}, whereas Section~\ref{sec:simulations_CPD} presents simulation results on  testing for a change point in the marginal distribution.   Finally, finite sample performance of our methodology in testing for change point and testing for independence between two series in the context of functional data is examined in Section~\ref{sec:simulations_FDA} and \ref{sec:sim_Indep_FDA}, respectively. In all these subsections, we denote our proposed methodology by "SS-SN" (Sample Splitting \& Self-Normalization).

\subsection{Goodness-of-Fit Testing}\label{sec:simulations_GOF}

In this subsection, we investigate the empirical performance of our proposed methodology in the context of goodness-of-fit testing. In particular, we generate data from the same AR$(1)$ model defined in ~\eqref{eqn:data_null_CPD}. We are interested in testing the null hypothesis $H_0:P=\cN(0,1/(1-\kappa^2))$ at $5\%$ significance level, where $\kappa\in \{-0.5, 0, 0.5\}$ is assumed to be known. The sample splitting ratio $\eta$ in our proposed method is chosen from the grid of values $\{0.2, 0.35, 0.5,  0.7\}.$ We work with the Gaussian kernel $K(x,y)=\exp{(-\sigma(x-y)^2)}$, where $\sigma$ is chosen via the "Median Heuristics", i.e.,
\[\sigma=Median\left\{\frac{1}{2|Y_{i_1}-Y_{i_2}|^2}, 1\leq i_1\neq i_2 \leq n\right\}.\]
According to \cite{kellner2019one}, we have the following formulae for the mean embedding of $P$ into $\mathcal{H}(K)$ and its RKHS norm, when $P=\cN(0,1/(1-\kappa^2))$:
$$\mu_P(x)=\left(1+\frac{2\sigma}{1-\kappa^2}\right)^{-1/2}\exp{\left(-\frac{\sigma\,x^{\top}x}{ \left(1+\frac{2\sigma}{1-\kappa^2}\right)}\right)};\;\;\;\left\|\mu_{P}\right\|_{\mathcal{H}(K)}^2=\left(1+\frac{4\sigma}{1-\kappa^2}\right)^{-\frac{1}{2}}.$$ 
We compare the performance of our method against the fixed-$b$ subsampling method by \cite{shao2013fixed}. We briefly describe their methodology for the sake of completeness.  
 Let $F_{a:b}(x):=\frac{1}{b-a+1}\sum_{i=a}^{b} 1(Y_i\leq x)$ and let $F_0$ be the pre-specified distribution under the null hypothesis. The p-value based on fixed-$b$ subsampling is defined as:
$$
pval_{n, l}^*=\frac{1}{N} \sum_{t=1}^N \mathbf{1}\left\{\sqrt{l}\left\|F_{t: t+l-1}(x)-F_{1: n}(x)\right\|_{\infty} \geq \sqrt{n}\left\|F_{1: n}(x)-F_0\left(x \right)\right\|_{\infty}\right\}.
$$
where $\|\cdot\|_{\infty}$ is the infinity norm, $l=\lfloor n b\rfloor, N=n-l+1$ and $b \in(0,1)$ is the blocklength
parameter. The distribution of $pval_{n, l}^*$ is not pivotal and depends on b, but can be approximated by a second-level subsampling. Let $n^{\prime}$ be the second-level subsampling length, $l^{\prime}=\max \left(\left\lceil n^{\prime} b\right\rceil, 2\right)$ and $N^{\prime}=n^{\prime}-l^{\prime}+1$. For $t=1,2, \ldots, n-n^{\prime}+1$, define
$$
h_{n^{\prime}, t}^*=\frac{1}{N^{\prime}} \sum_{j=t}^{t+N^{\prime}-1} 1\left\{\sqrt{l^{\prime}}\left\|F_{j: j+l^{\prime}-1}(x)-F_{t: t+n^{\prime}-1}(x)\right\|_{\infty} \geq \sqrt{n^{\prime}}\left\|F_{t: t+n^{\prime}-1}(x)-F_{1:n}\left(x\right)\right\|_{\infty}\right\}.
$$

Then the final p -value is defined as
$$
pval^*=\frac{1}{n-n^{\prime}+1} \sum_{t=1}^{n-n^{\prime}+1} 1\left\{h_{n^{\prime}, t}^*<\text { pval }_{n, l}^*\right\}.
$$
 Then, the null hypothesis is rejected at level $\alpha$ if $pval^*\leq \alpha.$ We let $b\in \{0.05, 0.15, 0.25\}$ and set the second-level subsampling length $n'=30.$ The experiment is repeated over $5000$ independent Monte Carlo replicates and the average rejection rates in percentage (and run times in milliseconds) are tabulated in Table \ref{tab:Size_GOF}. 
While the computational complexity of both methods seems comparable,
the subsampling method faces serious size distortion, especially for higher values of $b$. In contrast, the SS-SN methodology has excellent size accuracy across all the values of the sample splitting ratio. This indicates that the size accuracy of our proposed methodology is not sensitive to the choice of the splitting ratio.\\
\begin{table*}
\centering
    \caption{Empirical Size in \% (average run time in milliseconds) in Goodness-of-Fit Testing}
    \label{tab:Size_GOF}
    \begin{tabular}{@{}l|cccc|ccc@{}}
        \hline
     & \multicolumn{4}{c|}{SS-SN} & \multicolumn{3}{c}{Fixed-b}\\
   \hline
        $\kappa$ & $\eta=0.2$ & $\eta=0.35$ & $\eta=0.5$ & $\eta=0.7$ & $b=0.05$ & $b=0.15$ & $b=0.25$\\
        \hline 
       -0.5 & 5.18 (3.67)& 4.94 (3.99)& 4.38 (4.07)& 5.28  (3.81)& 7.52 (10.64)&10.88 (10.11)&13.56 (9.2)\\ 
        0 & 4.88 (3.64)& 5.1 (3.98)& 4.86 (4.08)& 5.1 (3.1)& 8.78 (10.76)& 10.12 (10.18)&13.96 (9.23)\\ 
        0.5 & 5.5 (2.92)&5.54 (3.2)&6.28 (3.29)&5.74 (3.1)&4.54 (10.84)&9.56 (9.99)& 14.22 (9.09)\\ \hline
    \end{tabular}
\end{table*}

\begin{figure}[t]
    \centering

    \begin{subfigure}[t]{7.5cm}
        \includegraphics[width=7.5cm]{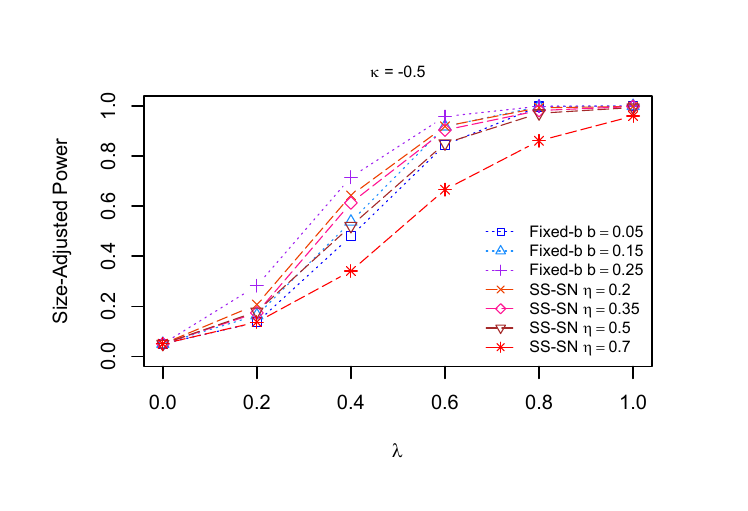}
        \caption{}
    \end{subfigure}
    \hspace{1cm}
    \begin{subfigure}[t]{7.5cm}
        \includegraphics[width=7.5cm]{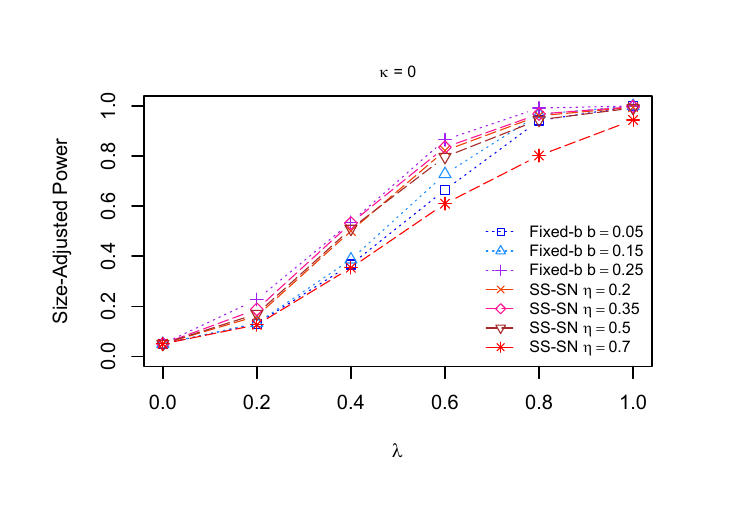}
        \caption{}
    \end{subfigure}

    \vspace{1em} 

    \begin{subfigure}[t]{7.5cm}
        \includegraphics[width=7.5cm]{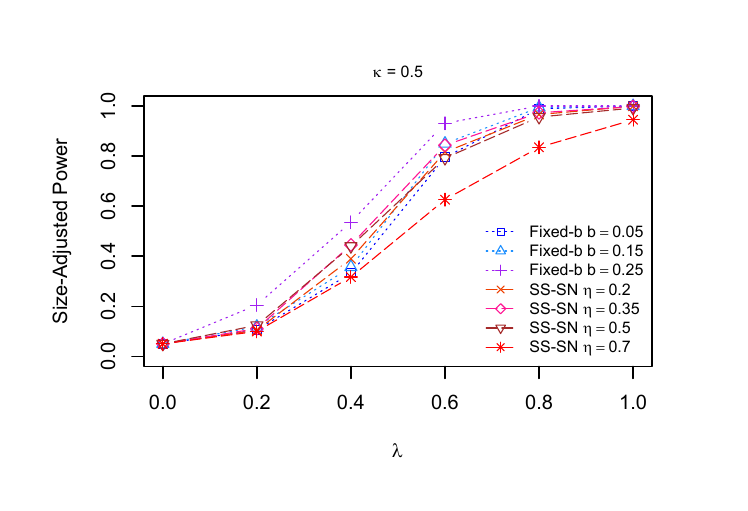}
        \caption{}
    \end{subfigure}

    \caption{Size-adjusted power in goodness-of-fit testing}
    \label{fig:power_GOF}
\end{figure}
We now analyze the empirical size-adjusted power of the SS-SN and fixed-$b$ subsampling method. We generate the data from 
     $$X_t=(1-\delta_t) Y_{t}+\delta_t \epsilon_t;\;t=1,2,\cdots,n,$$
where $Y_t=\kappa Y_{t-1}+e_t$, $e_t$ s are iid $N(0, 1)$, $\epsilon_t$s are iid $Exp(1)-1$ and $\delta_t$ s are iid Bernoulli$(\lambda)$ with $\lambda\in \{0, 0.2, 0.4,\cdots,1\} $. Moreover, $\{e_t\}_{t\in \mathbb{Z}}$, $\{\epsilon_t\}_{t\in \mathbb{Z}}$ and $\{\delta_t\}_{t \in \mathbb{Z}}$ s are mutually independent.
The values of all other parameters involved remain the same as in the size analysis. We plot the size-adjusted power of different methods against $\lambda$ in Figure \ref{fig:power_GOF}.
In general, the SS-SN method exhibits very mild power loss in comparison to the fixed-$b$ subsampling method. However, the power loss seems to be significant for the splitting ratio $\eta=0.7$. This is consistent with the finding in \cite{zhang2024another} and suggests that we should set $\eta\le 0.5$ to prevent potential power loss.

\subsection{Testing for a Change Point in Euclidean Data}\label{sec:simulations_CPD}

In this subsection, we investigate the finite sample performance of our proposed change point testing in marginal distribution against that of the sequential block bootstrap method proposed by \cite{sharipov2016sequential}. To investigate the finite sample size, we generate data from the following univariate AR$(1)$ model:
\begin{align}\label{eqn:data_null_CPD}
    Y_t=\kappa Y_{t-1}+\epsilon_t;\;t=1,2,\cdots, 200,
\end{align}
where $\epsilon_t$ s are iid $\mathcal{N}(0,1)$ and $\kappa \in \{-0.5, 0, 0.5\}$.
In the implementation of our methodology, we choose the splitting parameter $\eta=0.1$ and  work with Gaussian and Student kernels, which are given by
$$K_1(x,y)=\exp{(-\sigma(x-y)^2)},\;\;K_2(x,y)=(1+\sigma(x-y)^2)^{-1},$$
respectively. The bandwidth of the kernels $\sigma$ is chosen via the "Median Heuristics", i.e.,
$$\sigma=\operatorname{Median}\left\{\frac{1}{2|Y_{i_1}-Y_{i_2}|^2}, 1\leq i_1\neq i_2 \leq n\right\}.$$ To implement the sequential block bootstrap method of \cite{sharipov2016sequential}, we set the weight function $w(t)$ to be Gaussian, i.e., $w(t)=\frac{1}{\sqrt{2\pi}}\exp{(-x^2/2)}$, and the blocklength is chosen to be $5$ or $10$ (denoted in Figure~\ref{fig:power_CPD} as "SBB-b5" and "SBB-b10", respectively). The number of bootstrap replicates is chosen to be $499$. The entire experiment is repeated over $1000$ independent Monte Carlo replicates. The average rejection rates in percentage (and run times in seconds) are tabulated in Table \ref{tab:Size_CPD}, which shows that 
 our proposed test has excellent size accuracy across all values of $\kappa$ and there seems to be no apparent effect from the choice of the kernel. In contrast, the sequential block bootstrap method is over-conservative for negative $\kappa$ and exhibits noticeable size distortion when $\kappa$ is positive. Another noticeable advantage of our method is  in computational time. 
 While our method can be executed in less than 50 miliseconds on average, the sequential block bootstrap method can take up to about half an hour. 
\begin{table*}
\centering
    \caption{Empirical Size in \% (average run time in seconds) in Change Point Testing}
    \label{tab:Size_CPD}
    \begin{tabular}{@{}l|cc|cc@{}}
         \hline
     & \multicolumn{2}{c}{SS-SN} & \multicolumn{2}{c}{Sequential Block Bootstrap}\\
    \hline 
        $\kappa$ & Gaussian Kernel & Student Kernel & Blocklength=5 & Blocklength=10\\
        \hline 
       -0.5 & 5.4 (0.04)& 5.5 (0.03)& 2.4 (850.04)& 3.5 (502.47)\\ 
        0 & 5.9 (0.04)& 5.7 (0.04)& 4.9 (1137.12)& 3.5 (963.69)\\ 
        0.5 & 5.6 (0.04)& 6  (0.04)& 10.4 (1036.67)& 7.1 (1509.97)\\ \hline
    \end{tabular}
\end{table*}

To investigate the size adjusted power of the above methods, we generate data from the following model, where there is a single change point for the marginal distribution in the middle of the sequence:
$$X_t=(1-\delta_t) Y_{t}+\delta_t Y_{t} 1_{t\leq \lfloor n/2\rfloor}+\delta_t \epsilon_t 1_{t> \lfloor n/2\rfloor};\;t=1,2,\cdots,n,$$
where $Y_t=\kappa Y_{t-1}+e_t$, $e_t$s are iid $\cN(0,1)$,   $\epsilon_t$s  are iid $Exp(1)-1$ and $\delta_t$s are iid Bernoulli$(\lambda)$. Moreover, $\{e_t\}_{t\in \mathbb{Z}}$, $\{\epsilon_t\}_{t\in \mathbb{Z}}$ and $\{\delta_t\}_{t \in \mathbb{Z}}$ s are mutually independent. We choose\\ $\lambda \in \{0,\,0.2,\,0.4,\,\cdots, 1\}$. The choice of $\kappa$, the splitting ratio $\eta$, the kernels and their bandwidths in our method, and the blocklengths and weight function in the sequential block bootstrap method remain the same as our choices in the finite sample size analysis. However, we reduce the number of bootstrap replicates to 299 in this case in the interest of computational time. To perform size adjustment for the bootstrap-based test in \cite{sharipov2016sequential}, we follow the methodology proposed by \cite{dominguez2000size}. The experiment is repeated over $1000$ independent Monte Carlo replicates and we plot the size-adjusted power of different methods against $\lambda$ in Figure \ref{fig:power_CPD}. 
\begin{figure}[t]
    \centering

    \begin{subfigure}[t]{7.5cm}
        \includegraphics[width=7.5cm]{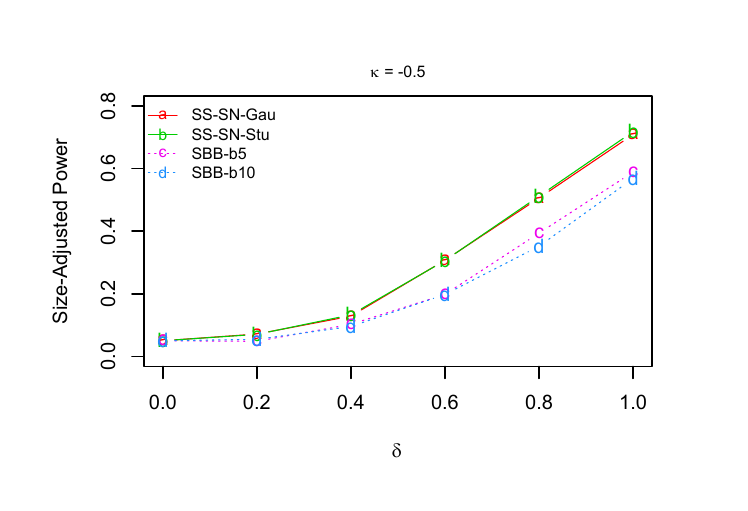}
        \caption{}
    \end{subfigure}
    \hspace{1cm}
    \begin{subfigure}[t]{7.5cm}
        \includegraphics[width=7.5cm]{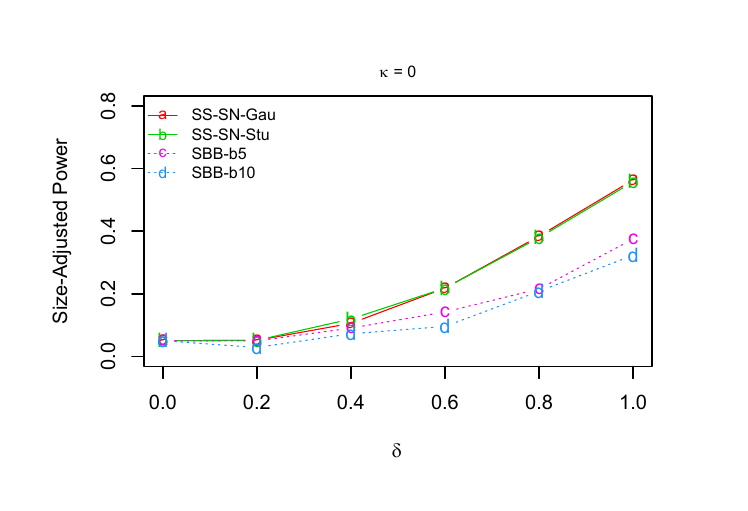}
        \caption{}
    \end{subfigure}


    \begin{subfigure}[t]{7.5cm}
        \includegraphics[width=7.5cm]{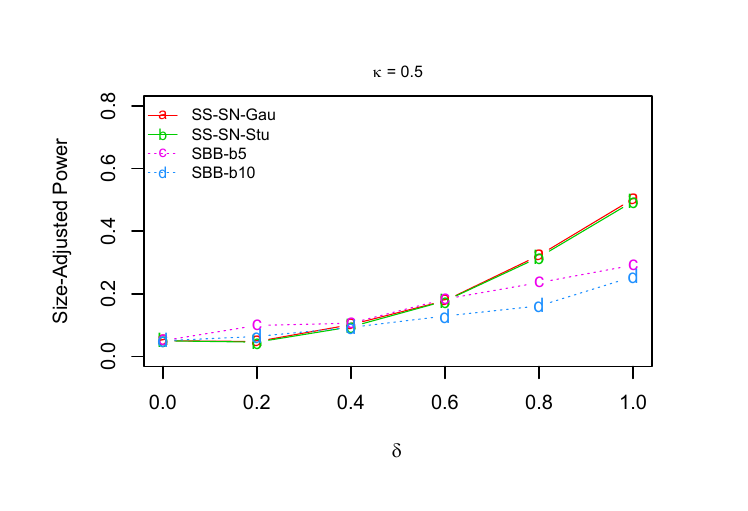}
        \caption{}
            \end{subfigure}
    \caption{  Size-adjusted power in change point testing in Euclidean data}
    \label{fig:power_CPD}
    \end{figure}
    
\noindent As one can see, the power of our SS-SN method is competitive, and it outperforms the sequential block bootstrap most of the time. In addition, the two kernels deliver almost identical powers, suggesting potential insensitivity of our method with respect to the kernel choice.

\subsection{Testing for a Change Point in Functional Time Series}\label{sec:simulations_FDA}
In this subsection, we investigate the finite sample performance of our proposed method in testing for a single change point in the marginal distribution of a functional time series. We generate the data under the null hypothesis from two types of processes. Under the first data-generating process, observations are generated from functional AR(1) process where the innovations are standard Brownian motions. We use an approximation on a evenly spaced finite grid of $d$ points $\{i/d:\; 0\leq i \leq d-1\}.$ In particular,
$$X_{-50}=\left(\xi_1, \xi_1+\xi_2, \ldots, \sum_{i=1}^d \xi_i\right) / \sqrt{d}, \quad \xi_i\text{s are i.i.d. }\mathcal{N}(0,1)\text{-distributed;}
$$
$$
X_t= \Phi X_{t-1}^{\mathrm{T}}+W_t \quad \forall-50<t \leq n,
$$
where $\Phi \in \mathbb{R}^{d \times d}$ with entries $\Phi_{i, j}=\min (i, j) / d^2,$ and 
\begin{equation}\label{eqn:FDA_DGP1}
    W_t=\left(\xi_1^{(t)}, \xi_1^{(t)}+\xi_2^{(t)}, \ldots, \sum_{i=1}^d \xi_i^{(t)}\right) / \sqrt{d},
\end{equation}

$\xi_i^{(t)}$s are i.i.d. $\mathcal{N}(0,1)$-distributed. To see the effect of heavy-tailed distribution, we consider another data-generating process where $\xi_i^{(t)}$s for $-50<t\leq 200$ in the above model are instead iid sample from a $t$ distribution with $2$ degrees of freedom. We set $d=100$ in all our simulations.  We compare the empirical rejection rates of our method against that of the spatial sign test of \cite{wegner2024robust} when $n=50,\; 200$ at four different levels of significance, i.e., $\alpha=0.1, 0.05, 0.025, 0.01$. For the implementation of both methods, we approximate the Hilbert space distance between any two observations in $\mathcal{L}_2[0,1]$ by $\sqrt{\frac{1}{d}\sum_{k=0}^{d-1} (X_i(k/d)-X_j(k/d))^2}.$ The splitting ratio in SS-SN is set to be $\eta=0.1$. Both the kernels in SS-SN are chosen to be Gaussian and the kernel bandwidth is chosen via median heuristics. In the spatial sign test of \cite{wegner2024robust}, the bandwidth $q$ involved in the multiplier bootstrap is selected via the data adaptive approach of \cite{rice2017plug}. The number of Monte Carlo replicates is set to be $1000$ and the number of bootstrap replicates is set to be $500$. As shown in Table~\ref{tab:Size_CPD_FDA}, the SS-SN method is very slightly oversized for small sample size, but the size accuracy gets better for large sample size. Moreover, the size accuracy is good even in heavy-tailed data. In contrary, the spatial sign test is significantly undersized for small sample size. In terms of computational time, the SS-SN method took $0.4$ seconds on average per iteration, whereas the spatial sign test took $1.55$ hours.\\
\begin{table*}
\centering
    \caption{Empirical Size in \% for Change Point Testing in Functional Data}
    \label{tab:Size_CPD_FDA}
    \begin{tabular}{@{}l|cccc|cccc@{}}
         \hline
     & \multicolumn{4}{c|}{Gaussian} & \multicolumn{4}{c}{Heavy tail}\\
    \hline 
    & \multicolumn{2}{c}{SS-SN} & \multicolumn{2}{c}{Spatial Sign} & \multicolumn{2}{c}{SS-SN} & \multicolumn{2}{c}{Spatial Sign}\\
    \hline
        $\alpha$ & $n=50$ & $n=200$ &  $n=50$ & $n=200$ & $n=50$ & $n=200$ & $n=50$ & $n=200$ \\
        \hline
      0.1 & 11.50 & 9.80 & 0.50 & 13.55 & 10.70 & 9.70 & 0.40 & 14.95 \\ 
  0.05 & 6.60 & 4.60 & 0.05 & 5.45 & 4.75 & 4.20 & 0.10 & 6.90 \\ 
  0.025 & 3.85 & 2.25 & 0.05 & 2.05 & 3.05 & 1.80 & 0.05 & 3.70 \\ 
  0.01 & 2.20 & 0.65 & 0.05 & 0.50 & 1.70 & 0.60 & 0.05 & 1.25 \\ 
        \hline
    \end{tabular}
\end{table*}

To examine the power performance, we consider a uniform jump of magnitude $\lambda$ in the mean of the first data-generating process above in the middle of the sequence, i.e.,
\begin{align*}
        Y_t=\begin{cases}
            X_t & t\leq \lfloor n/2 \rfloor,\\
            X_t+\lambda\,(1,1,\cdots,1)^{\top} & t>\lfloor n/2 \rfloor,
        \end{cases}
        \end{align*}
where $\{X_t\}$s are defined in \eqref{eqn:FDA_DGP1}. We set $n=200$ here and the number of Monte-Carlo replicates to be $500$. We vary $\lambda\in\{0,\;0.2,\cdots,1\}$ and plot the size-adjusted power of aforementioned methods against $\lambda$ in Figure~\ref{fig:power_CPD_FDA}. Note that we choose the splitting ratio $\eta$ of the SS-SN method from the set $\{0.05, 0.1, 0.15\}$. Overall, the power of SS-SN is comparable to that of the spatial sign test and in fact, it is slightly better when the splitting ratio is $0.1$ or $0.15$.
\begin{figure}[t]
\centering
\includegraphics{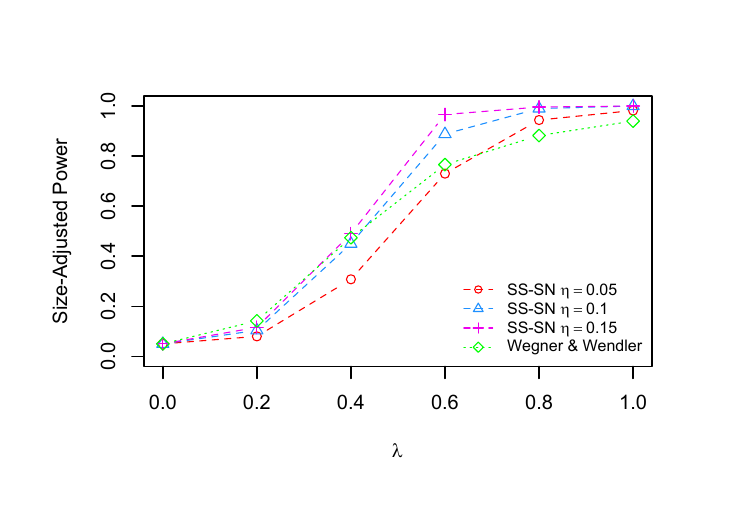}
\caption{ Size-adjusted power in change point testing in functional data}\label{fig:power_CPD_FDA}
\end{figure}
\subsection{Testing for Independence between Two Functional Time Series}\label{sec:sim_Indep_FDA}
In this subsection, we analyze the finite sample performance of our proposed test for independence between two stationary functional time series. We compare the SS-SN method against the method proposed by \cite{horvath2015testing}. Under the null hypothesis, we generate samples from two independent stationary functional autoregressive model of order one, i.e. for $1\leq i \leq 200$ and $t \in [0,1]$,
\begin{equation}\label{eqn:FDA_DGP2}
   X_i(t)=\int \psi_q(t,u)X_{i-1}(u)du+B_{1,i}(t),\;\;\;Y_i(t)=\int \psi_q(t,u)Y_{i-1}(u)du+B_{2,i}(t), 
\end{equation}
where $\psi_q(t,u)=q\,\min(t,u)$ and $\{B_{1,i}\}_{i=1}^{200} \indep \{B_{2,i}\}_{i=1}^{200}$ are standard Brownian motions on $[0,1]$. The parameter $q$ determines the magnitude of temporal dependence and we let $q\in \{0, 0.75, 2\}$. Note that the data-generating process is exactly the same as ~\eqref{eqn:FDA_DGP1} when $q=1$; When $q=0$, the two samples correspond to  independent standard Brownian motions on $[0,1].$ 
To implement the SS-SN method, we choose both kernels to be Gaussian and the kernel bandwidth is once again selected via median heuristics. To investigate the effect due to the choice of the splitting ratio $\eta$ and the lag $m$ associated with our test statistic, we chose $\eta$ from $\{0.2, 0.35, 0.5\}$ and $m$ from $\{0,1,2\}.$ In the implementation of the method of \cite{horvath2015testing}, we choose both kernels to be Bartlett kernels and the associated windows were set to be $w_1(r)=w_2(r)=\lfloor r^{1/4}\rfloor.$ Inspired by the numerical experiments in their paper, we set the associated lag $H$ to be $4$. The empirical rejection rates based on 2000 Monte Carlo replicates are tabulated in Table~\ref{tab:Size_Indep_FDA}. The sizes for both method appear accurate when $q=0, 0.75$, whereas the method of \cite{horvath2015testing} exhibits non-negligible size distortion for stronger temporal dependence (i.e., $q=2$). The SS-SN method exhibits overall good size accuracy with $\eta=0.2$, and the size distortion can be noted for higher values of $m$ and when $\eta=0.35, 0.5$. The choice of $m$ can affect the power (to be shown below) and may be based on domain knowledge. \\
\begin{table}
\centering
\caption{Empirical Size in \% for Testing Independence between Two Functional Series}
    \label{tab:Size_Indep_FDA}
    \begin{tabular}{@{}r|ccc|ccc|ccc|c@{}}
    \hline
     & \multicolumn{9}{|c|}{SS-SN} & Horvath \& Rice\\
    \hline 
      & \multicolumn{3}{|c|}{$\eta=0.2$} &\multicolumn{3}{|c|}{$\eta=0.35$} & \multicolumn{3}{|c|}{$\eta=0.5$} & \\
    \hline
    & $m=0$ & $m=1$ & $m=2$ & $m=0$ & $m=1$ & $m=2$ & $m=0$ & $m=1$ & $m=2$ & $H=6$\\
    \hline
    $q=0$ & 6.5 & 5.5 & 5.9  & 5.8 & 5.3 & 7.2 & 5.6 & 5.7 & 6.1 & 5.6\\
    $q=0.75$ & 5 & 5.35 & 6.2 & 5.6 & 6.3 & 5.9 & 5.3 & 5.6 & 5.9 & 6\\
    $q=2$ & 5.5 & 5.7 & 6.2 & 6.6 & 7 & 6.8 & 6.7 & 7.5& 8.4 & 11\\
    \hline
    \end{tabular}
\end{table}
To investigate the size-adjusted power, we use innovations $\{W_{1,i}\}_{i=1}^{200}$ and $\{W_{2,i}\}_{i=1}^{200}$ in place of $\{B_{1,i}\}_{i=1}^{200} \text{ and } \{B_{2,i}\}_{i=1}^{200}$  in \eqref{eqn:FDA_DGP2}, respectively, where $W_{1,i}(t):=B_{1,i}(t)$ and\\ $W_{2,i}(t):=\lambda\, B_{1,i}(t)+\sqrt{1-\lambda^2}\,B_{2,i}(t)$, $\lambda \in \{0, 0.1, 0.2, \cdots,1\}.$ We set the number of Monte-Carlo replicates to be $1000$ and the average size-adjusted power is plotted against $\lambda$ in Figure~\ref{fig:Power_Indep_FDA}. The choice of $\eta$ does not seem to significantly affect the power of the SS-SN test, although the SS-SN test with lag $m=0$ seem to  outperform SS-SN tests with $m=1,2$. This is expected because the cross-correlation between the two samples under the alternative hypothesis is non-zero at lag $0$, but zero at higher lags. Overall, the size-adjusted power of SS-SN method seems to be comparable with \cite{horvath2015testing} when the dependence is not too strong. However, when the dependence is strong, SS-SN method results in slightly more power loss.
    \begin{figure}[t]
    \centering
    \begin{subfigure}[t]{5cm}
        \includegraphics[width=5cm]{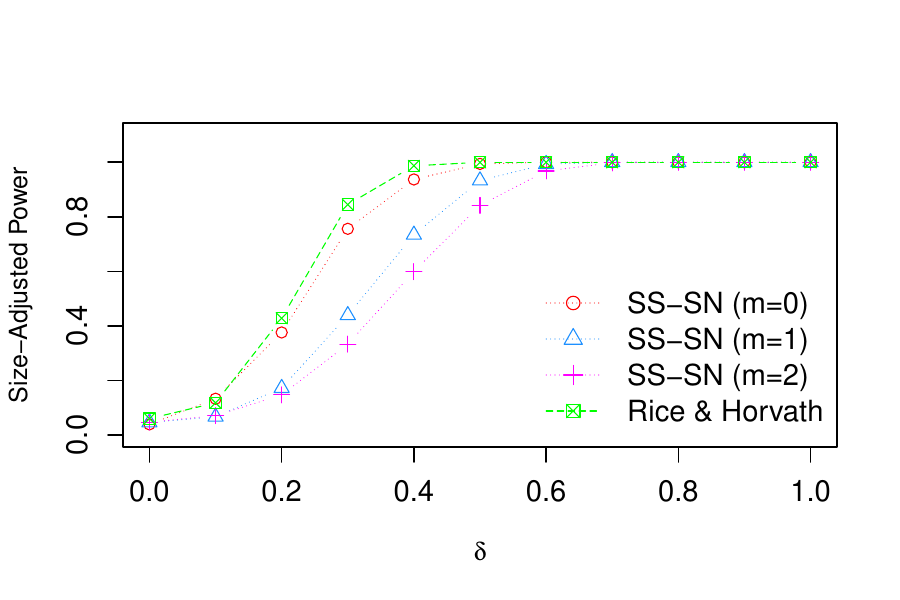}
        \caption{$q=0,\,\eta=0.2$}
    \end{subfigure}
    \hfill
    \begin{subfigure}[t]{5cm}
        \includegraphics[width=5cm]{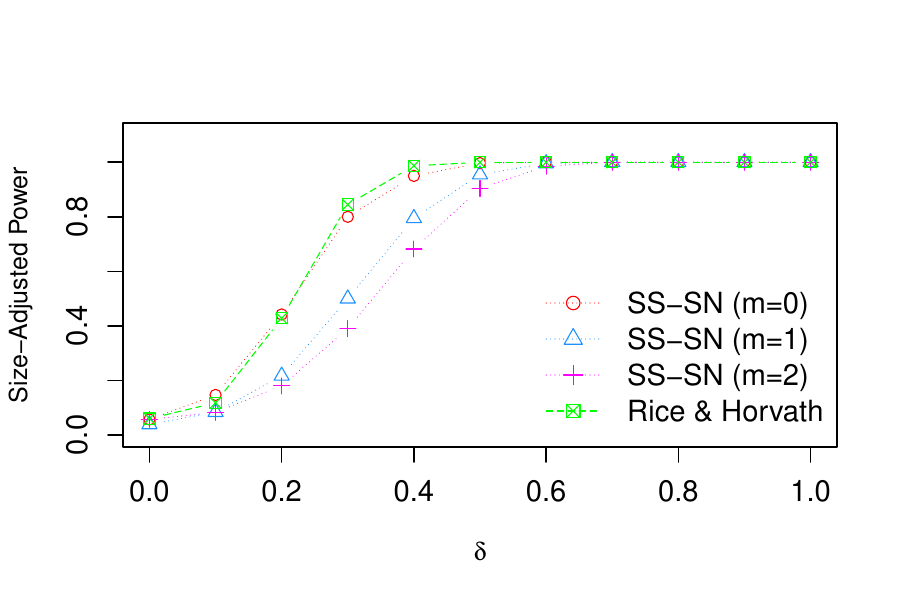}
        \caption{$q=0.75,\,\eta=0.2$}
    \end{subfigure}
    \hfill
    \begin{subfigure}[t]{5cm}
        \includegraphics[width=5cm]{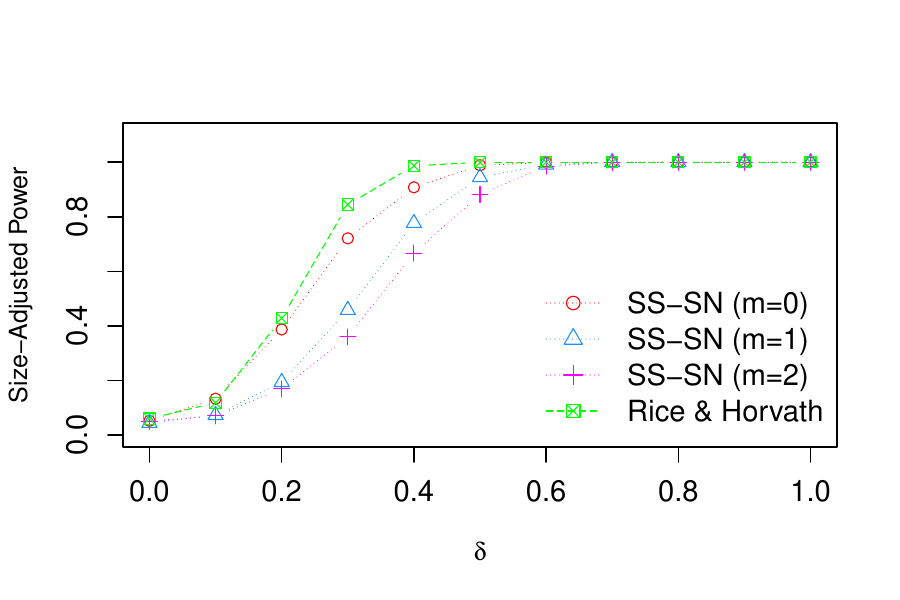}
        \caption{$q=2,\,\eta=0.2$}
    \end{subfigure}
    \hfill
    \begin{subfigure}[t]{5cm}
        \includegraphics[width=5cm]{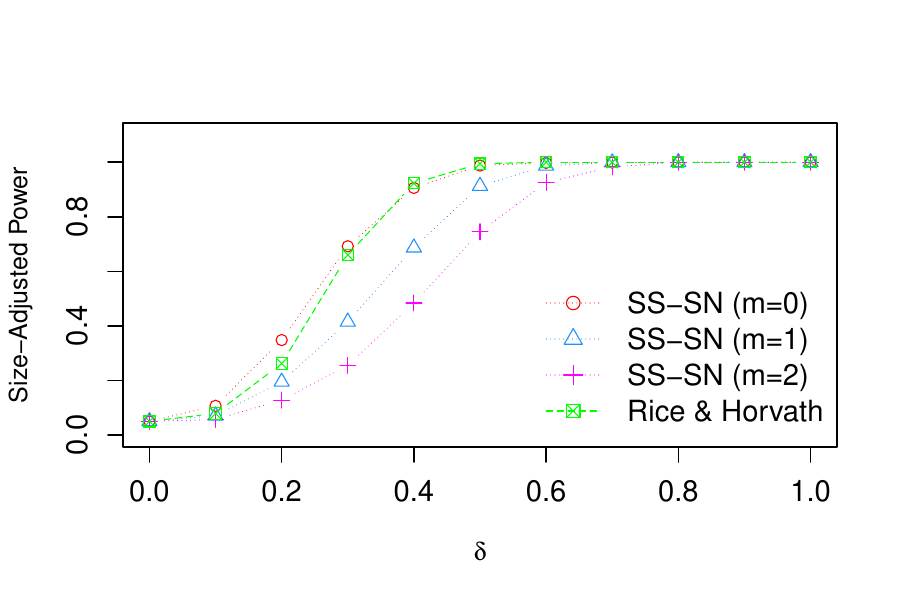}
        \caption{$q=0,\,\eta=0.2$}
    \end{subfigure}
    \hfill
    \begin{subfigure}[t]{5cm}
        \includegraphics[width=5cm]{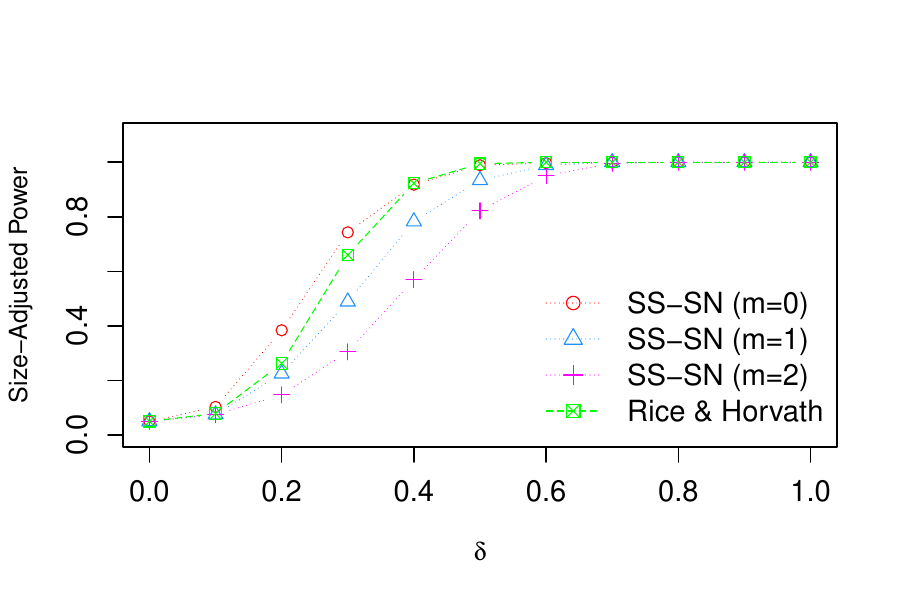}
        \caption{$q=0,\,\eta=0.2$}
    \end{subfigure}
    \hfill
    \begin{subfigure}[t]{5cm}
        \includegraphics[width=5cm]{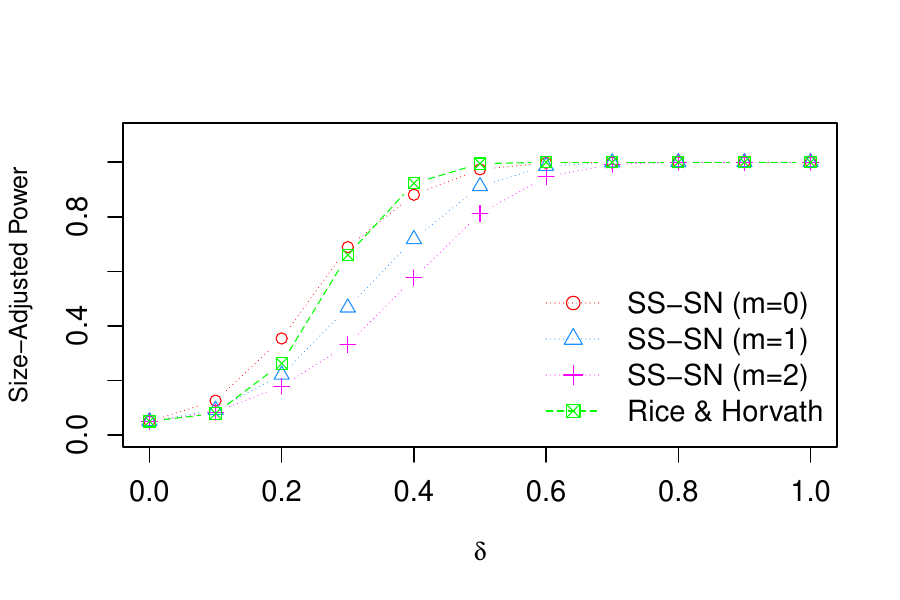}
        \caption{$q=0,\,\eta=0.2$}
    \end{subfigure}
    \hfill
    \begin{subfigure}[t]{5cm}
        \includegraphics[width=5cm]{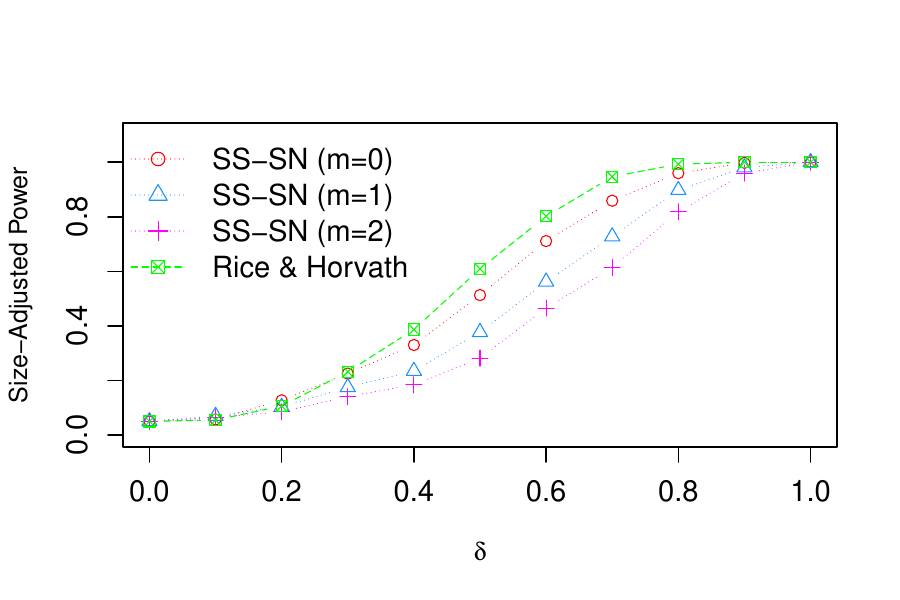}
        \caption{$q=0,\,\eta=0.2$}
    \end{subfigure}
    \hfill
    \begin{subfigure}[t]{5cm}
        \includegraphics[width=5cm]{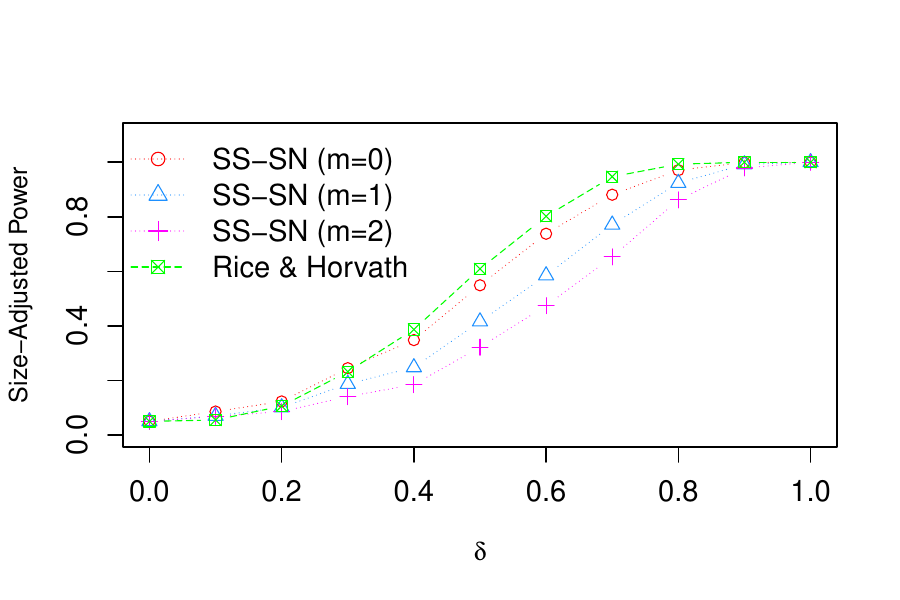}
        \caption{$q=0,\,\eta=0.2$}
    \end{subfigure}
    \hfill
    \begin{subfigure}[t]{5cm}
        \includegraphics[width=5cm]{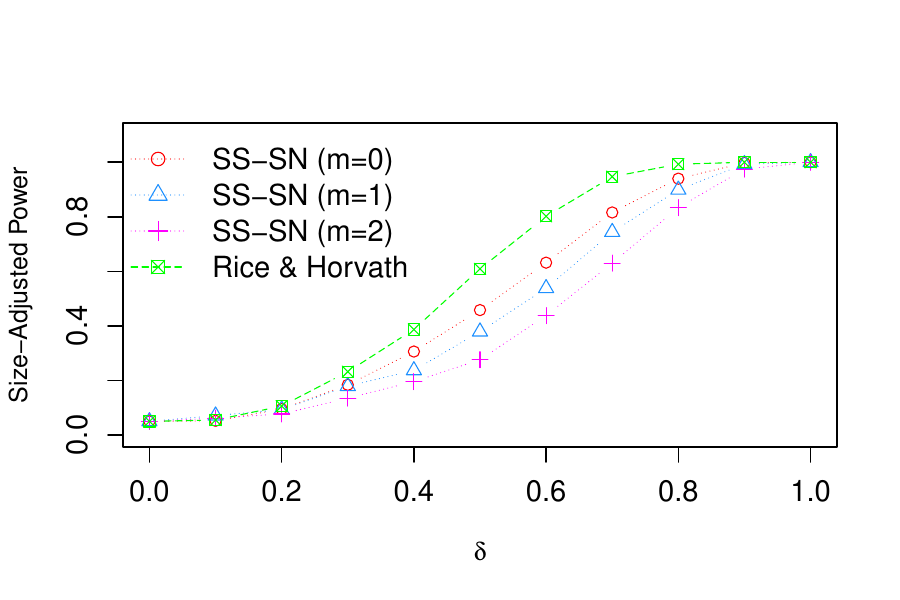}
        \caption{$q=0,\,\eta=0.2$}
    \end{subfigure}
    \caption{Size-adjusted power in testing independence between two functional series}
    \label{fig:Power_Indep_FDA}
\end{figure}

 \section{Conclusion}\label{sec:conclusion}

In this article, we propose a new resampling-free methodology to conduct inference for various infinite-dimensional parameters in multivariate, Hilbert space valued  and object-valued time series. With the help of sample-splitting and RKHS embedding, we extend the use of self-normalization, a bandwidth-free method typically used for drawing inference about finite-dimensional parameters in time series, to the case of infinite-dimensional parameters. The proposed methodology is very easy to implement and fairly fast to compute. Theoretically, we derive the asymptotic pivotal distributions of our test statistics under the null and the limiting power under alternatives with mild moment and weak dependence conditions. 
Our methodology exhibits excellent size accuracy in finite samples with occasional mild to moderate power loss as compared to several existing competitors. 
Although the sample splitting ratio is left up to the choice of the user, we empirically demonstrate that this tuning parameter has little impact on the size accuracy or efficiency of our tests as long as it is in a suitable range, see Section~\ref{sec:simulations}. 
The size accuracy, computational efficiency and reduced sensitivity with respect to the tuning parameter are the major advantages of our method compared to the traditional bootstrap-based counterparts. One possible future research topic is to extend our methodology to high-dimensional or non-stationary time series setup. This is currently under investigation. 

\section*{Acknowledgements}
The authors would like to thank Dr. Gregory Rice for providing the code used in \cite{horvath2015testing} and some related clarifications. The partial support from a grant provided by the U.S. National Science Foundation is gratefully acknowledged.

\bigskip
    
	\setlength{\bibsep}{0pt plus 0.25ex}
	
	\bibliographystyle{chicago}

    \bibliography{corrected_mybibilography}
    
\begin{appendix}
    \refstepcounter{section}
\section*{Appendix \thesection: Gaussian Random Element and Brownian Motion in Hilbert Space}\label{appna}
\begin{definition}\label{defn:gaussian_hilbert}
    Let $H$ be a separable Hilbert space equipped with inner product $\langle .,.\rangle_H$. An $H$-valued random element $N$ is said to be Gaussian if for all $h\in H\setminus \{0\}$, $\langle N,h\rangle_H$ has a normal distribution in $\mathbb{R}.$ The distribution of $N$ is uniquely determined by its mean and covariance operator.
\end{definition}
Let us represent the space of all $H$-valued continuous and cadlag functions on $[0,1]$ by $C_H[0,1]$ and $D_H[0,1]$, respectively.  For $C_H[0,1]$, we use the $\cL_{\infty}$ metric. For $D_H[0,1]$ , we use the Skorokhod metric on $D_H[0,1]$ defined in \cite{sharipov2016sequential} (also see \cite{billingsley2013convergence}, Chapter 12), i.e.,
$$
\rho(f, g)=\inf _{\lambda \in \Lambda}\left\{\sup _{t \in[0,1]}\|f(t)-g \circ \lambda(t)\|_H+\|i d-\lambda\|_{\infty}\right\},\;\; f, g \in D_H[0,1],
$$
where $\Lambda$ is the class of strictly increasing, continuous mappings of $[0,1]$ onto itself, $\|\cdot\|_H$ is the Hilbert space norm, and $\|\cdot\|_{\infty}$ is the supremum norm. Moreover $i d:[0,1] \rightarrow[0,1]$ is the identity function and $\circ$ denotes the composition of functions. Note that the Skorokhod metric is dominated by the $\cL_{\infty}$ metric (by taking $\lambda=id$).
\begin{definition}\label{defn:brownian_hilbert}
     A random element $W$ of $D_H[0,1]$ is called Brownian motion in $H$ if
     \begin{itemize}
\item[(i)] $W(0)=0$ almost surely;
\item[(ii)] $W \in C_H[0,1]$ almost surely, where $C_H[0,1]$ is the set of all continuous functions from $[0,1]$ to $H$;
\item[(iii)] the increments on disjoint intervals are independent;
\item[(iv)] for all $0 \leq t<t+s \leq 1$ the increment $W(t+s)-W(t)$ is Gaussian with mean zero and covariance operator $s S$, where $S: H \rightarrow H$ does not depend on $s$ or $t$.
\end{itemize}
Note that the distribution of a Brownian motion $W$ is uniquely determined by the covariance operator $S$ of $W(1)$.
\end{definition}

\refstepcounter{section}
\section*{Appendix \thesection: Basics of HSIC}\label{appnb}
\begin{definition}\label{defn:HS_norm}
     Let $H_1,\,H_2$ be two separable Hilbert spaces. The tensor product of two elements $f$ and $g$ from Hilbert spaces $H_1$ and $H_2$, respectively, is a linear map $f \otimes g: H_2\mapsto H_1$, defined as
$$(f\otimes g)h=f\langle g,h\rangle,\;\;\forall\;\;h\in H_2.$$
The tensor product Hilbert space associated with $H_1$ and $H_2$ is the space
$$H_1 \otimes H_2=\{f \otimes g\mid f\in H_1,g \in H_2\},$$
equipped with the inner product
$$\langle f_1\otimes g_1, f_2 \otimes g_2\rangle_{H_1 \otimes H_2}=\langle f_1,f_2\rangle_{H_1}\times \langle g_1,g_2\rangle_{H_2}.$$
The norm associated with the inner product $\langle .,. \rangle_{H_1 \otimes H_2}$ is known as the Hilbert Schmidt norm and we denote it by $\left\|.\right\|_{HS}.$
\end{definition}
\begin{remark}
    Note that if $H_1$ and $H_2$ are separable, then $H_1\otimes H_2$ is also separable. 
\end{remark}
\begin{definition}
    Let $\mathcal{X},\mathcal{Y}$ be two separable metric spaces, $(X,Y)$ be two jointly distributed random elements such that $X\in \mathcal{X}$ and $Y\in \mathcal{Y}$. Given two positive-definite and symmetric kernels $K$ and $L$ on $\mathcal{X}$ and $\mathcal{Y}$, respectively, the associated cross-covariance operator $C_{XY}$ is defined as the Hilbert-space mean of the random operator $\left(K(X,.)-\mu_X\right)\otimes \left(L(Y,.)-\mu_Y\right)$, where $\mu_X:=\bE_X K(X,.)$ and $\mu_Y:=\bE_Y K(Y,.).$ This is equivalent to
$$C_{XY}=\bE_{XY}\left[\left(K(X,.)-\mu_X\right)\otimes \left(L(Y,.)-\mu_Y\right)\right],$$
where the above expectation is taken on the tensor product Hilbert space $\mathcal{H}(K)\otimes \mathcal{H}(L).$
\end{definition}
The following result tells us that under some mild conditions, a necessary and sufficient condition for the independence of $X$ and $Y$ is that the cross-covariance operator $C_{XY}$ is the $\mathbf{0}$ element of the tensor product Hilbert space, see \cite{sejdinovic2013equivalence}.
\begin{lemma}\label{lemma:HSIC=0_iff_indep}
If $\bE K(X,X),\;\;\bE L(Y,Y)$ are finite and both the kernels $K$ and $L$ are characteristic, then
$X\indep Y$ if and only if $C_{XY}=\mathbf{0}\in \mathcal{H}(K)\otimes \mathcal{H}(L).$
\end{lemma}
Motivated by the above characterization, the squared Hilbert Schmidt norm of the cross-covariance operator has been used a measure of dependence between $X$ and $Y$ (see \cite{gretton2005measuring} for more details). This is precisely the HSIC between $X$ and $Y$.
\begin{definition}
    The Hilbert-Schmidt Independence Criterion (HSIC) between $X$ and $Y$ is defined as $HSIC(X,Y):=\left\|C_{XY}\right\|^2_{HS}.$
\end{definition}
\refstepcounter{section}
\section*{Appendix \thesection: Functional Central Limit Theorem in Hilbert Space}\label{appnc}
In this section, we state the details of Theorem 4.6. in \cite{chen1998central}, which proves a functional central limit theorem (FCLT) in a separable Hilbert space. This result is the foundation of Proposition 1-3 in our paper. To this end, let $H$ be a separable Hilbert space equipped with inner product $\langle .,.\rangle_H$. Let $\{W_{n,i};\;n\geq 1;\;i=\cdots,-1,0,1,\cdots\}$ be a triangular array of $H$-valued random elements on a common probability space $(\Omega,\mathcal{F},\bP)$. Moreover, let $\{V_{n,i};\;n\geq 1;\;i=\cdots,-1,0,1,\cdots\}$ be another triangular array of random elements on $(\Omega,\mathcal{F},\bP)$ taking values in a Banach space $B$.  Let us define the following partial sum process in $D_H[0,1]$:
   \[M_n(t):=\sum_{i=1}^{\lfloor nt \rfloor} W_{n,i} ,\;t\in [0,1].\]
   Moreover, let us define the following process in $C_H[0,1]$:
   \[M_n^*(t)=M_n(t)+(nt-\lfloor nt\rfloor)W_{n,\lfloor nt\rfloor+1},\;t\in [0,1].\]   
   We now revisit the concepts of mixingales and near epoch dependence, which will be necessary to study the FCLT of the process $M^*_n$.
\begin{definition}[Mixingale]

    Let $\left\{\mathcal{F}^{n, i}: i=\ldots,-1,0,1, \ldots ; n \geq 1\right\}$ be an array of sub $\sigma$-fields of $\mathcal{F}$ non-decreasing in $i$ for each $n$,  defined by
    \[\mathcal{F}^{n,i}:=\sigma\{\cdots,V_{n,i-1},V_{n,i}\}.\]
     Then $\left\{W_{n, i}, \mathcal{F}^{n, i}\right\}$ is called an $\cL_{\boldsymbol{p}}(\mathbb{H})$ mixingale array if there exist finite nonnegative constants $\left\{c_{n, i}:\; i \geq 1 ; n \geq 1\right\}$ and $\left\{\psi_m ;\; m \geq 0\right\}$ with $\psi_m \rightarrow 0$ as $m \rightarrow \infty$ such that for all $i, n \geq 1, m \geq 0$ :
$$
\begin{aligned}
\left\|\bE\left(W_{n, i} \mid \mathcal{F}^{n, i-m}\right)\right\|_p & \leq \psi_m c_{n, i} ; \\
\left\|W_{n, i}-\bE\left(W_{n, i} \mid \mathcal{F}^{n, i+m}\right)\right\|_p & \leq \psi_{m+1} c_{n, i} .
\end{aligned}
$$

If in addition, for each $n \geq 1, W_{n, i}$ is $\mathcal{F}^{n, i}$-measurable, then $\left\{W_{n, i}, \mathcal{F}^{n, i}\right\}$ is an adapted $\cL_p(\mathbb{H})$-mixingale array, and the second inequality holds automatically.
\end{definition}

\begin{definition}[$\cL_p(H)$-Near Epoch Dependence]
    $\{W_{n,i};\;n,i\geq 1\}$ is called an $\cL_p(\mathbb{H})$-array near epoch dependent $\left(\cL_p(\mathbb{H})\right.$-NED $)$ on $\left\{V_{n, i}\right\}$ if $\left\|W_{n, i}\right\|_p<\infty$ for all $i, n \geq 1$, and there exist constants $\left\{d_{n, i} \geq 0: i \geq 1 ; n \geq 1\right\}$ and $\left\{\mu_m \geq 0: m \geq 0\right\}$ with $\mu_m$ decreasing to zero as $m \rightarrow \infty$ such that
$$
\left\|W_{n, i}-\bE\left[W_{n, i} \mid \mathcal{F}_{n, i-m}^{n, i+m}\right]\right\|_p \leq \mu_m d_{n, i},
$$
where $\mathcal{F}_{n,a}^{n,b}:=\sigma\{V_{n,a},V_{n,a+1},\cdots,V_{n,b}\}.$

\end{definition}
\begin{remark}
   It should be noted that when $W_{n,i}$ is $\mathcal{F}^{n,i}$-measurable, the above inequality holds trivially with $\mu_m=d_{n,i}=0$.
\end{remark}
\begin{definition}[Strong Mixing]
    The array $\{V_{n,i}\}$ is called a "strong" or "\(\alpha\)-mixing" sequence if $\lim_{m \to \infty}\alpha(m)=0$, where
    \[\alpha(m):=\sup_{n \geq 1}\,\sup_i\,\sup \left[|\bP(A \cap B)-\bP(A)\bP(B)|:\;A \in \mathcal{F}_{n,-\infty}^{n,i},\,\mathcal{F}_{n,i+m}^{n,\infty}\right].\]
\end{definition}
Now, we state the required FCLT in the following lemma, which is a slight modification of Theorem 4.6. in \cite{chen1998central}.
\begin{lemma}[Theorem 4.6. of \cite{chen1998central}]\label{lemma:fclt_chen_white}
   Let $\left\{W_{n, i}: i \geq 1 ; n \geq 1\right\}$ be a double array of $H$-valued random elements (r.e.'s) with zero means, $\left\|W_{n, i}\right\|_r<\infty$ for some $r \geq 2$.
   Suppose the following conditions hold for some $\delta>0$:
   \begin{itemize}
       \item [(a)]
       \begin{itemize}
           \item[(i)] $\{W_{n,i}\}$ is $\mathcal{L}_2(H)$-NED on $\{V_{n,i}\}$ such that $\mu_m=0$ for large enough $m$ and $d_{n,i}=0$ for $i\geq K$ for some fixed $K$. Moreover, $\left\{V_{n, i}\right\}$ is a strong mixing $B$-r.e. array with mixing coefficients $\alpha(m)$ such that $\alpha(m)=O(m^{-\lambda})$ for some $\lambda>1+\frac{2}{\delta}>0$, for some $\delta>0$;
           \item [(ii)] $\sup _{k, J} \limsup \operatorname{sum}_{n \rightarrow \infty} k^{-1} \sum_{J+1 \leq i \leq J+k}\left(c_{n, i}\right)^2<\infty$, where $c_{n, i} \equiv \max \left(d_{n, i},\left\|W_{n, i}\right\|_r\right)$ for all $n, i$;
       \end{itemize}
       \item [(b)] when $r=2,\left\{\left\|W_{n, i} / c_{n, i}\right\|^2: i \geq 1 ; n \geq 1\right\}$ is uniformly integrable for each $0 \neq h \in \mathbb{H}$.
       \item [(c)] for all $t \in[0,1]$, each $0 \neq h \in \mathbb{H}, \lim _{n \rightarrow \infty} E\left[\Sigma_{1 \leq i \leq \lfloor nt \rfloor}\langle W_{n, i}, h\rangle\right]^2=t \sigma^2(h)$.
       \item [(d)] $\lim \sup _{n \rightarrow \infty} E\left[\left\|\Sigma_{1 \leq i \leq n} W_{n, i}\right\|^2\right]=C<\infty$.
       \item [(e)]There exists a complete orthonormal basis $\left\{e_l\right\}$ such that
$$
\limsup _{n \rightarrow \infty} \Sigma_{l>k} \Sigma_{1 \leq i \leq n}\left(c_{n, i, e_l}\right)^2 \rightarrow 0, \quad \text { as } k \rightarrow \infty,
$$
where for any $h \in H$, $\{\langle W_{n,i},h \rangle\}$ is a $\cL_2(\mathbb{R})$-mixingale with $\{\psi_m\}$ and $\{c_{n,i,h}\}$.
   \end{itemize}
Then, (i) there exists nonsingular covariance operator $S$ such that $(S h, h)=\sigma^2(h)$ for each $h \in H, h \neq 0$; and (ii) $M^*_n \leadsto W$ in $C_{H}[0,1]$ (as $n \rightarrow \infty$ ), where $W$ is a Brownian motion in $H$ with $W(0)=0, \mathbb{E} W(1)=0$, and $\operatorname{Cov} W(1)=S$.
\end{lemma}
\begin{remark}
    Note that the above result provides FCLT for $M_n^*$ in $C_H[0,1]$ with respect to the $\cL_{\infty}$ metric. However, in our proof of Proposition 1, we show that using Lemma~\ref{lemma:asymp_negligble_process}, this indeed implies that weak convergence of the partial sum process $M_n$ (elements in $D_H[0,1]$) also holds in $D_H[0,1]$ with respect to the Skorokhod metric. 
\end{remark}
\refstepcounter{section}
\section*{Appendix \thesection: Proofs of Main Results}\label{appnd}
\subsection{Proof of Theorem 3.2.}
We define the following process in $[\eta,1]$:
    \begin{align*}
        S_n(r)&:=\frac{1}{n}\sum_{j=m_1+1}^{\lfloor nr \rfloor}\sum_{i=1}^{m_1} \langle K(Y_i,.)-\mu_{P_0}, K(Y_j,.)-\mu_{P_0}\rangle\\       &=\left\langle \frac{1}{\sqrt{n}}\sum_{i=1}^{m_1}(K(Y_i,.)-\mu_{P_0}),  \frac{1}{\sqrt{n}}\sum_{j=m_1+1}^{\lfloor nr \rfloor}(K(Y_j,.)-\mu_{P_0})\right\rangle,\;\; r \in[\eta,1].
    \end{align*}
An application of continuous mapping theorem on Assumption 1. gives us that as $n \to \infty$,
    $$\left\{S_n(r)\right\}_{r \in [\eta,1]}\leadsto \left\{\langle W(\eta), W(r)-W(\eta)\rangle\right\}_{r \in [\eta,1]} \text{ in } D[\eta,1].$$
    Now,
\begin{align*}
  &\sqrt{m_2}\cdot U_n\\
  &=\frac{S_n(1)}{\sqrt{\frac{1}{m_2}\sum_{t=m_1+1}^n\left(S_n(\frac{t}{n})-\frac{t-m_1}{m_2}S_n(1)\right)^2}}\\
  &=\frac{S_n(1)}{\sqrt{\frac{1}{1-\eta}\int_{\eta}^{1}\,\left(S_n(r')-\frac{r'-\eta}{1-\eta}S_n(1)\right)^2 dr'+o_P(1)}}\\
  &\overset{r=r'-\eta}{=}\frac{S_n(1)}{\sqrt{\frac{1}{1-\eta}\int_{0}^{1-\eta}\,\left(S_n(r+\eta)-\frac{r}{1-\eta}S_n(1)\right)^2 dr+o_P(1)}}\allowdisplaybreaks
  \\
  &\overset{D}{\to}\frac{\langle W(\eta), W(1)-W(\eta)\rangle}{\sqrt{\frac{1}{1-\eta}\int_{0}^{1-\eta}\,\Big(\langle W(\eta), W(r+\eta)-W(\eta)\rangle
  -\frac{r}{1-\eta} \langle W(\eta), W(1)-W(\eta)\rangle\Big)^2 dr}}\\
  &\equiv g\left(\{W(r)\}_{r \in [0,1]}\right),
\end{align*}
where the convergence in distribution in the penultimate line holds due to another application of the Continuous Mapping theorem. By the independent increment property of $W$, $\left\{W(t)\right\}_{t \in [0,\eta]}$ is independent of $\left\{W(r+\eta)-W(\eta)\right\}_{r \in [0,1-\eta]}$. Therefore, conditioning on\\ $W(\eta)=u$, we have
\begin{align*}
    &g\left(\{W(r)\}_{r \in [0,1]}\right) \mid \left\{W(\eta)=u\right\}\\
    &\overset{D}{=}\frac{\langle u, W(1)-W(\eta)\rangle}{\sqrt{\frac{1}{1-\eta}\int_{0}^{1-\eta}\,\Big(\langle u, W(r+\eta)-W(\eta)\rangle
  -\frac{r}{1-\eta} \langle u, W(1)-W(\eta)\rangle\Big)^2 dr}}\\
  &=\frac{\langle Su, u \rangle^{-1/2} \langle u, W(1)-W(\eta)\rangle}{\frac{\langle Su, u \rangle^{-1/2}}{\sqrt{1-\eta}}\sqrt{\int_{0}^{1-\eta}\,\Big(\langle u, W(r+\eta)-W(\eta)\rangle
  -\frac{r}{1-\eta} \langle u, W(1)-W(\eta)\rangle\Big)^2 dr}}\\
  &\overset{D}{=}\frac{b(1)}{\sqrt{\int_{0}^{1} \left(b(s(1-\eta)+\eta)-sb(1)\right)^2ds}},
\end{align*}
where $b(t)=B(t)-B(\eta)$, $t \in [\eta,1]$ and the final equality follows from Lemma \ref{lemma:distn_inner.prod_BM} and Continuous Mapping theorem.  Note that the conditional distribution doesn't depend on $u$ and hence, we conclude that it will also be the unconditional distribution. Now, if we let $h\left(B(r)\right)=\frac{1}{\sqrt{1-\eta}}\, b(r(1-\eta)+\eta,1);\;\; r \in [0,1]$, then $\left(\{h(B(r))\}_{r \in [0,1]}\right)\overset{D}{=}\{B(r)\}_{r \in [0,1]}.$ Therefore,
\begin{align*}
g\left(\{B(r)\}_{r \in [0,1]}\right) &\overset{D}{=}
   \frac{h(B(1))}{\sqrt{\int_{0}^{1}\, \left(h(B(s))-sh(B(1))\right)^2\, ds}}\\
   &\overset{D}{=}\frac{B(1)}{\sqrt{\int_{0}^{1}\, \left(B(s)-sB(1)\right)^2\, ds}}=U.
\end{align*}
\subsection{Proof of Theorem 3.2.}
 Let us define $\tilde{Z}_i=K(Y_i,.)-\mu_{P_n}$.From the definition of $\left\{S_n(r)\right\}_{r \in [\eta,1]}$, it follows that:
     \begin{align*}
        S_n(r)
        &=\left\langle \frac{1}{\sqrt{n}}\sum_{i=1}^{m_1}(Z_i-\mu_{P_0}),  \frac{1}{\sqrt{n}}\sum_{j=m_1+1}^{\lfloor nr \rfloor}(Z_j-\mu_{P_0})\right\rangle\\
        &=\left\langle \frac{1}{\sqrt{n}}\sum_{i=1}^{m_1}(Z_i-\mu_{P_n}+\mu_{P_n}-\mu_{P_0}),  \frac{1}{\sqrt{n}}\sum_{j=m_1+1}^{\lfloor nr \rfloor}(Z_j-\mu_{P_n}+\mu_{P_n}-\mu_{P_0})\right\rangle\\
        &=\frac{1}{n}\sum_{j=m_1+1}^{\lfloor nr \rfloor}\sum_{i=1}^{m_1} \langle \tilde{Z}_i,\tilde{Z}_j\rangle-\frac{1}{n}\sum_{j=m_1+1}^{\lfloor nr \rfloor}\sum_{i=1}^{m_1}\langle \tilde{Z}_i, \mu_{P_n}-\mu_{P_0}\rangle\\
        &-\frac{1}{n}\sum_{j=m_1+1}^{\lfloor nr \rfloor}\sum_{i=1}^{m_1}\langle \tilde{Z}_j, \mu_{P_n}-\mu_{P_0}\rangle +\frac{(\lfloor nr \rfloor-m_1) m_1}{n}\left\|\mu_{P_n}-\mu_{P_0}\right\|^2\allowdisplaybreaks\\
        &=\frac{1}{n}\sum_{j=m_1+1}^{\lfloor nr \rfloor}\sum_{i=1}^{m_1} \langle \tilde{Z}_i,\tilde{Z}_j\rangle-\frac{1}{n}\sum_{j=m_1+1}^{\lfloor nr \rfloor}\sum_{i=1}^{m_1}\langle \tilde{Z}_i,\delta_n\rangle-\frac{1}{n}\sum_{j=m_1+1}^{\lfloor nr \rfloor}\sum_{i=1}^{m_1}\langle \tilde{Z}_j, \delta_n\rangle\\
        &+\frac{(\lfloor nr \rfloor-m_1) m_1}{n}\left\|\delta_n\right\|^2.
    \end{align*}
    
    \begin{enumerate}
        \item $\Delta_n^2\to 0$: This means $\sqrt{n} \delta_n \to \mathbf{0}$, where $\mathbf{0}$ represents the zero element of $\mathcal{H}(K)$. 
        Since $(\tilde{Z}_t)_{t=1}^{n}$ satisfy Assumption 2, therefore as $n \to \infty,$
   \begin{align*}
       \sup_{r \in [\eta,1]}\frac{1}{n}\sum_{j=m_1+1}^{\lfloor nr \rfloor}\sum_{i=1}^{m_1}\langle \tilde{Z}_i,\delta_n\rangle&=\sup_{r \in [\eta,1]}\langle \frac{1}{\sqrt{n}}\sum_{i=1}^{m_1}Z_i, \frac{1}{\sqrt{n}}\sum_{j=m_1+1}^{\lfloor nr \rfloor} \delta_n \rangle\\
       &\leq \left\|\frac{1}{\sqrt{n}}\sum_{i=1}^{m_1}Z_i \right\|\cdot \sup_{r \in [\eta,1]} \frac{\lfloor nr \rfloor -m_1}{\sqrt{n}}\left\|\delta_n\right\|\overset{\bP}{\to} 0.
   \end{align*}
    Similarly, 
    $$\sup_{r \in [\eta,1]}\frac{1}{n}\sum_{j=m_1+1}^{\lfloor nr \rfloor}\sum_{i=1}^{m_1}\langle \tilde{Z}_j, \delta_n\rangle \overset{\bP}{\to} 0.$$
    Finally, as a consequence of Assumption 2 and continuous mapping theorem, we have
    \begin{align*}
        \frac{1}{n}\sum_{j=m_1+1}^{\lfloor nr \rfloor}\sum_{i=1}^{m_1} \langle \tilde{Z}_i,\tilde{Z}_j\rangle\leadsto \left\{\langle \tilde{W}(\eta), \tilde{W}(r)-\tilde{W}(\eta)\rangle\right\}_{r \in [\eta,1]} \text{ in } D[\eta,1].
    \end{align*}
    Therefore,
    $$S_n(r)\leadsto \left\{\langle \tilde{W}(\eta), \tilde{W}(r)-\tilde{W}(\eta)\rangle\right\}_{r \in [\eta,1]} \text{ in } D[\eta,1].$$
    Observe that this limiting process is same as the limiting process of $S_n(r)$ under $H_0.$ Therefore, using the same conditioning technique that we used under the null, we conclude that under the alternative, $\sqrt{m_2}\cdot U_n \overset{D}{\to} U,$ as $n \to \infty.$ Therefore, $\bP(\sqrt{m_2}\cdot U_n>U_{1-\alpha})\to \alpha.$
    \item $\Delta_n^2 \to \infty:$ Observe that in this case,
    \begin{align*}
        \frac{1}{\Delta_n^2}S_n(1)\overset{D}{\to}\; \eta (1-\eta),
    \end{align*}
    and hence,
    $$S_n(1)=O_{\bP}(\Delta_n^2).$$
    Also,
    \begin{align*}
        &\frac{1}{m_2}\sum_{t=m_1+1}^n\left(S_n(\frac{t}{n})-\frac{t-m_1}{m_2}S_n(1)\right)^2\\\
        &=\frac{1}{n^2 m_2}\sum_{t=m_1+1}^n\Bigg(\sum_{j=m_1+1}^{t}\sum_{i=1}^{m_1} \langle \tilde{Z}_i,\tilde{Z}_j\rangle-\sum_{j=m_1+1}^{t}\sum_{i=1}^{m_1}\langle \tilde{Z}_i,\delta_n\rangle-\sum_{j=m_1+1}^{t}\sum_{i=1}^{m_1}\langle \tilde{Z}_j, \delta_n\rangle \\
        &+(t-m_1) m_1\left\|\delta_n\right\|^2\\\
        &-\frac{t-m_1}{m_2}\Bigg(\sum_{j=m_1+1}^{n}\sum_{i=1}^{m_1} \langle \tilde{Z}_i,\tilde{Z}_j\rangle-\sum_{j=m_1+1}^{n}\sum_{i=1}^{m_1}\langle \tilde{Z}_i,\delta_n\rangle\\
        &-\sum_{j=m_1+1}^{n}\sum_{i=1}^{m_1}\langle \tilde{Z}_j, \delta_n\rangle +m_2 m_1\left\|\delta_n\right\|^2\Bigg)\Bigg)^2\allowdisplaybreaks\\
        &=\frac{1}{n^2 m_2} \sum_{t=m_1+1}^n\Bigg( \underbrace{\left(\sum_{j=m_1+1}^{n}\sum_{i=1}^{m_1} \langle \tilde{Z}_i,\tilde{Z}_j\rangle-\frac{t-m_1}{m_2}\sum_{j=m_1+1}^{n}\sum_{i=1}^{m_1} \langle \tilde{Z}_i,\tilde{Z}_j\rangle\right)}_{=A_n(t/n)}\\
        &+\underbrace{m_1 \left\langle \delta_n, \left(\sum_{j=m_1+1}^{t}\tilde{Z}_j-\frac{t-m_1}{m_2}\sum_{j=m_1+1}^{n}\tilde{Z}_j\right)\right\rangle}_{=B_n(t/n)}\Bigg)^2\leq \frac{2}{n^2 m_2}\sum_{t=m_1+1}^n (A^2_n(t/n)+B^2_n(t/n)).
            \end{align*}
    It's not difficult to see that under Assumption 2, a direct application of continuous mapping theorem gives us
    \begin{align*}
        \frac{1}{n^3 m_2\left\|\delta_n\right\|^2}\sum_{t=m_1+1}^nA^2_n(t/n)=O_{\bP}(1),
    \end{align*}
    and 
    \begin{align*}
        \frac{1}{n^3 m_2\left\|\delta_n\right\|^2}\sum_{t=m_1+1}^nB^2_n(t/n)=O_{\bP}(1).
    \end{align*}
 Therefore,
    $$\frac{1}{m_2}\sum_{t=m_1+1}^n\left(S_n(\frac{t}{n})-\frac{t-m_1}{m_2}S_n(1)\right)^2=O_P(\Delta_n^2).$$
    Hence,
    \begin{align*}
        \sqrt{m_2}\cdot U_n =\frac{S_n(1)}{\sqrt{\frac{1}{m_2}\sum_{t=m_1+1}^n\left(S_n(\frac{t}{n})-\frac{t-m_1}{m_2}S_n(1)\right)^2}}\overset{\bP}{\to} \infty,
    \end{align*}
   which implies that $\bP(\sqrt{m_2}\cdot U_n>U_{1-\alpha})\to 1.$
    \item $\Delta_n^2 \to c^2$, where $c \in (0,\infty)$ and, $\frac{\delta_n}{\left\|\delta_n\right\|}\to \delta_0 \in \mathcal{H}(K)$: In this case $\sqrt{n}\delta_n \to c \delta_0$. By continuous mapping theorem,
    \begin{align*}
        &S_n(r)
        \\
        &=\frac{1}{n}\sum_{j=m_1+1}^{\lfloor nr \rfloor}\sum_{i=1}^{m_1} \langle \tilde{Z}_i,\tilde{Z}_j\rangle-\frac{1}{n}\sum_{j=m_1+1}^{\lfloor nr \rfloor}\sum_{i=1}^{m_1}\langle \tilde{Z}_i,\delta_n\rangle-\frac{1}{n}\sum_{j=m_1+1}^{\lfloor nr \rfloor}\sum_{i=1}^{m_1}\langle \tilde{Z}_j, \delta_n\rangle\\
        &+\frac{(\lfloor nr \rfloor-m_1) m_1}{n}\left\|\delta_n\right\|^2\\
        &\leadsto \langle \tilde{W}(\eta), \tilde{W}(r)-\tilde{W}(\eta)\rangle+ c(r-\eta)\langle \delta_0 , \tilde{W}(\eta)\rangle + c \eta \langle \delta_0, \tilde{W}(r)-\tilde{W}(\eta)\rangle +c^2 \eta (r-\eta)\\
        &=:C(r),        
    \end{align*}
    in $D[\eta,1]$. Therefore, similar to the calculations in the proof of Theorem 3.1., by another application of continuous mapping theorem, we obtain that as $n \to \infty$,
    \begin{align*}
        \sqrt{m_2} \cdot U_n\overset{D}{\to} U^*:=\frac{C(1)}{\sqrt{\int_{0}^{1}\left(C(s(1-\eta)+\eta)-sC(1)\right)^2\,ds}}.
    \end{align*}
    Since $\left\{\tilde{W}(s)\right\}_{s \in [0,\eta]}$ is independent of $\left\{\tilde{W}(r)-\tilde{W}(\eta)\right\}_{r \in [\eta, 1]}$, conditioning on\\ $\left\{\tilde{W}(\eta)=u\right\}$, we obtain that
    \begin{align*}
        &\left\{C(r)\right\}_{r \in [\eta,1]}\mid \left\{\tilde{W}(\eta)=u\right\}\\
        &\overset{D}{=} \left\{\langle u+c \eta \delta_0, \tilde{W}(r)-\tilde{W}(\eta)\rangle +c (r-\eta) \langle \delta_0, u \rangle + c^2 \eta (r-\eta)\right\}_{r \in [\eta,1]}\\
        &\overset{D}{=}\left\{\underbrace{\sqrt{\langle \tilde{S}(u+c \eta \delta_0), u+c \eta \delta_0\rangle}}_{=:\Theta(u)}\left(B(r)-B(\eta)\right)+c (r-\eta) \langle \delta_0, u \rangle + c^2 \eta (r-\eta)\right\}_{r \in [\eta,1]} .
    \end{align*}
     Note that
    \begin{align*}
        &\left\{C(s(1-\eta)+\eta)-sC(1)\right\}_{s \in [0,1]}\\
        &\overset{D}{=}\Theta(u)\left\{\left(B(s(1-\eta)+\eta)-B(\eta)-s(B(1)-B(\eta))\right)\right\}_{s \in [0,1]}.
    \end{align*}
    Therefore, conditioned on $\left\{\tilde{W}(\eta)=u\right\}$, we have
    \begin{align*}
        U^*\mid \left\{W(\eta)=u\right\}&\overset{D}{=} \frac{(B(1)-B(\eta))+\frac{c (1-\eta) \langle \delta_0, u \rangle}{\Theta(u)}+\frac{c^2 \eta (1-\eta)}{\Theta(u)}}{\sqrt{\int_{0}^{1}\left(B(s(1-\eta)+\eta)-B(\eta)-s(B(1)-B(\eta))\right)^2 ds}}\allowdisplaybreaks\\
        &\overset{D}{=}\frac{B(1)+\frac{c (1-\eta) \langle \delta_0, u \rangle}{\Theta(u)}+\frac{c^2 \eta (1-\eta)}{\Theta(u)}}{\sqrt{\int_{0}^{1}\left(B(s)-sB(1)\right)^2 ds}},
    \end{align*}
    where in the last line we have used the fact that\\ $\frac{1}{\sqrt{1-\eta}}\left\{B(s(1-\eta)+\eta)-B(\eta)\right\}_{s \in [0,1]}\overset{D}{=}\left\{B(s)\right\}_{s \in [0,1]}.$ This completes the proof.
    \end{enumerate}
    \subsection{Proof of Theorem 3.3.}
    Let us define $Z_i'=K(Y_i,.)-\mu_P$, where $i=1,2,\cdots,n.$ Now, note that
\begin{align*}
    &N^{1/2}T_n(k)\\
    &= \sum_{t=1}^{k} \left(Z_t-\Bar{Z}_{1:N}\right)\\
    &=\frac{1}{m_1}\sum_{t=1}^{k}\sum_{i=1}^{m_1} K(Y_i, Y_{m_1+t})-\frac{1}{m_1}\sum_{t=1}^{k}\sum_{j=n-m_1+1}^{n} K(Y_j, Y_{m_1+t})-\frac{k}{N m_1}\sum_{t=1}^{N}\sum_{i=1}^{m_1} K(Y_i, Y_{m_1+t})\\
    &+\frac{k}{N m_1}\sum_{t=1}^{N}\sum_{j=n-m_1+1}^{n} K(Y_j, Y_{m_1+t})\allowdisplaybreaks\\
    &=\frac{1}{m_1}\sum_{t=1}^{k}\sum_{i=1}^{m_1} \left(K(Y_i, Y_{m_1+t})-\mu_P(Y_i)-\mu_P(Y_{m_1+t})+\left\|\mu_P\right\|^2\right)\\
    &-\frac{1}{m_1}\sum_{t=1}^{k}\sum_{j=n-m_1+1}^{n} \left(K(Y_j, Y_{m_1+t})-\mu_P(Y_j)-\mu_P(Y_{m_1+t})+\left\|\mu_P\right\|^2\right)\\
    &-\frac{k}{N m_1}\sum_{t=1}^{N}\sum_{i=1}^{m_1} \left(K(Y_i, Y_{m_1+t})-\mu_P(Y_i)-\mu_P(Y_{m_1+t})+\left\|\mu_P\right\|^2\right)\\
    &+\frac{k}{N m_1}\sum_{t=1}^{N}\sum_{j=n-m_1+1}^{n} \left(K(Y_j, Y_{m_1+t})-\mu_P(Y_j)-\mu_P(Y_{m_1+t})+\left\|\mu_P\right\|^2\right)\\
    &=\frac{1}{m_1}\sum_{t=1}^{k}\sum_{i=1}^{m_1} \langle Z'_i, Z'_{m_1+t} \rangle -\frac{1}{m_1}\sum_{t=1}^{k}\sum_{j=n-m_1+1}^{n} \langle Z'_j, Z'_{m_1+t}\rangle -\frac{k}{N m_1}\sum_{t=1}^{N}\sum_{i=1}^{m_1} \langle Z'_i, Z'_{m_1+t} \rangle\\
    &+ \frac{k}{N m_1}\sum_{t=1}^{N}\sum_{j=n-m_1+1}^{n} \langle Z'_j, Z'_{m_1+t}\rangle.
\end{align*}
With a little abuse of notation, setting $k=\lfloor nr \rfloor$, where $r \in [0,1-2\eta]$, applying Continuous Mapping theorem on Assumption 1., as $n \to \infty$, we have
\begin{align*}
    N^{1/2}T_n(k) &\leadsto \frac{1}{\eta} \langle W(\eta)-W(1)+W(1-\eta), W(r+\eta)-W(\eta) \rangle \\
    &-\frac{r}{\eta(1-2\eta)} \langle W(\eta)-W(1)+W(1-\eta), W(1-\eta)-W(\eta) \rangle\\
    &=\frac{1}{\eta} \langle W(\eta)-W(1)+W(1-\eta), W(r+\eta) -W(\eta) -\frac{r}{1-2\eta}(W(1-\eta)-W(\eta))\rangle.
\end{align*}
Similarly, as $n \to \infty$, we have
\begin{align*}
    &\frac{1}{N}\sum_{t=1}^k\left(Z_{1:t}-t\Bar{Z}_{1:k}\right)^2 \\
    &\leadsto \frac{1}{\eta^2(1-2\eta)}\int_{0}^{r} \Big(\langle W(\eta)-W(1)+W(1-\eta), W(s+\eta) -W(\eta)\rangle\\
    &-\frac{s}{r}\langle W(\eta)-W(1)+W(1-\eta), W(r+\eta) -W(\eta)\rangle\Big)^2\,ds\\
    &=\frac{1}{\eta^2(1-2\eta)}\int_{0}^{r} \langle W(\eta)-W(1)+W(1-\eta), W(s+\eta) -W(\eta)-\frac{s}{r}(W(r+\eta)-W(\eta))\rangle^2\,ds.
\end{align*}
\begin{align*}
    &\frac{1}{N}\sum_{t=k+1}^{N}\left(Z_{t:N}-(N-t+1)\Bar{Z}_{k+1:N}\right)^2\\
    &\leadsto \frac{1}{\eta^2(1-2\eta)}\int_{r}^{1-2\eta} \langle W(\eta)-W(1)+W(1-\eta),\\
    &W(1-\eta) -W(s+\eta)-\frac{1-2\eta-s}{1-2\eta-r}(W(1-\eta)-W(r+\eta))\rangle^2\,ds.
\end{align*}
Therefore, by Continuous Mapping Theorem, as $n \to \infty$, we have
\begin{align*}
    &G_n\\
    &\overset{D}{\to} \sup_{ r \in [0,1-2\eta]}\sqrt{1-2\eta}\langle W(\eta)-W(1)+W(1-\eta), W(r+\eta) -W(\eta) -\frac{r}{1-2\eta}(W(1-\eta)\\
    &-W(\eta))\rangle\Big/\Big(\int_{0}^{r} \langle W(\eta)-W(1)+W(1-\eta), W(s+\eta) -W(\eta)-\frac{s}{r}(W(r+\eta)-W(\eta))\rangle^2\,ds\\
    &+\int_{r}^{1-2\eta} \langle W(\eta)-W(1)+W(1-\eta), W(1-\eta) -W(s+\eta)\\
    &-\frac{1-2\eta-s}{1-2\eta-r}(W(1-\eta)-W(r+\eta))\rangle^2\,ds\Big)^{1/2}\\
    &=h\left(\left\{W(t)\right\}_{t \in [0,1]}\right).
\end{align*}
Now, due to the independent increment property of $W$, the process $\left\{W(r+\eta)-W(\eta)\right\}_{r \in [0,1-2\eta]}$ is independent of\newline \noindent $\left\{W(t)\right\}_{t \in [0,\eta]}\bigcup \left\{W(r+1-\eta)-W(1-\eta)\right\}_{r \in [0,\eta]}$. Thus, conditioning on $W(\eta)-W(1)+W(1-\eta)=u$ and using Lemma~\ref{lemma:distn_inner.prod_BM}, we have
\begin{align*}
    &h\left(\left\{W(t)\right\}_{t \in [0,1]}\right) \mid \left\{W(\eta)-W(1)+W(1-\eta)=u\right\}\\
    &\overset{D}{=}\sup_{ r \in [0,1-2\eta]}\sqrt{1-2\eta}\left(B(r+\eta) -B(\eta) -\frac{r}{1-2\eta}(B(1-\eta)-B(\eta))\right)\\
    &\Big/\sqrt{\begin{pmatrix}
        \int_{0}^{r}  \left(B(s+\eta) -B(\eta)-\frac{s}{r}(B(r+\eta)-B(\eta))\right)^2 ds\\
        +\int_{r}^{1-2\eta} \left( B(1-\eta) -B(s+\eta)-\frac{1-2\eta-s}{1-2\eta-r}(B(1-\eta)-B(r+\eta))\right)^2ds
    \end{pmatrix}}\allowdisplaybreaks\\
    &\overset{r'=\frac{r}{1-2\eta}}{=}\sup_{ r' \in [0,1]}\left(B(r'(1-2\eta)+\eta) -B(\eta) -r'(B(1-\eta)-B(\eta))\right)\\
    &\Bigg/\Big(\int_{0}^{(1-2\eta)r'}  \frac{1}{1-2\eta}\left(B(s+\eta) -B(\eta)-\frac{s}{(1-2\eta)r'}(B((1-2\eta)r'+\eta)-B(\eta))\right)^2 ds\\
    &+\int_{(1-2\eta)r'}^{1-2\eta} \frac{1}{1-2\eta}\Bigg( B(1-\eta) -B(s+\eta)-\frac{1-2\eta-s}{1-2\eta-(1-2\eta)r'}(B(1-\eta)\\
    &-B((1-2\eta)r'+\eta))\Bigg)^2ds\Big)^{1/2}\\
    &\overset{t=\frac{s}{1-2\eta}}{=}\sup_{ r' \in [0,1]}\left(B(r'(1-2\eta)+\eta) -B(\eta) -r'(B(1-\eta)-B(\eta))\right)\\
    &\Bigg/\Big(\int_{0}^{r'}  \left(B((1-2\eta)t+\eta) -B(\eta)-\frac{t}{r'}(B((1-2\eta)r'+\eta)-B(\eta))\right)^2 dt\\
    &+\int_{r'}^{1} \left( B(1-\eta) -B((1-2\eta)t+\eta)-\frac{1-t}{1-r'}(B(1-\eta)-B((1-2\eta)r'+\eta))\right)^2dt\Big)^{1/2}.\\
\end{align*}
Now, 
$$\frac{1}{\sqrt{1-2\eta}}\left\{B((1-2\eta)t+\eta) -B(\eta)\right\}_{t\in [0,1]}\overset{D}{=}\left\{B(s)\right\}_{s \in [0,1]}.$$
Therefore,
\begin{align*}
     &h\left(\left\{W(t)\right\}_{t \in [0,1]}\right) \mid \left\{W(\eta)+W(1)-W(1-\eta)=u\right\}\\
     &\overset{D}{=}\sup_{r'\in[0,1]}\;\frac{B(r')-rB(1)}{\left(\int_{0}^{r'}(B(s)-\frac{s}{r'}B(r'))^2 ds+\int_{r'}^{1}(B(1)-B(s)-\frac{1-s}{1-r'}(B(1)-B(r')))^2 ds\right)^{1/2}}.
\end{align*}
Since the conditional distribution doesn't depend on $u$, we conclude that it is also  the unconditional distribution, which completes the proof.
\subsection{Proof of Theorem 3.4.}
Let us redefine $Z_i'$ as $Z'_i:=K(Y_i,.)-\mu_{P_i}$ and let $\hat{\nu}_1:=\frac{1}{m_1}\sum_{i=1}^{m_1}Z_i'$ and $\hat{\nu}_n:=\frac{1}{m_1}\sum_{j=n-m_1+1}^{n}Z_j'$.
    Note that
    \begin{equation}\label{eqn:CPD_Power_1}
        \sqrt{n}(\hat{\nu}_1-\hat{\nu}_n-\delta_n)=\frac{n}{\lfloor n \eta \rfloor} \frac{1}{\sqrt{n}}\sum_{i=1}^{m_1}\left(Z'_i-Z'_{n-m_1+i}\right)-\sqrt{n}\,\delta_n.
    \end{equation}
    Therefore, for any $0\leq r \leq 1$, it follows that
    \begin{align}\label{eqn:CPD_Power_2}
        &\langle\hat{\nu}_1-\hat{\nu}_n-\delta_n, \sum_{j=1}^{\lfloor Nr \rfloor} Z'_{j+m_1}\rangle\\
        &=\frac{n}{\lfloor n \eta \rfloor}\langle\frac{1}{\sqrt{n}}\sum_{i=1}^{m_1}\left(Z'_i-Z'_{n-m_1+i}\right),\frac{1}{\sqrt{n}}\sum_{j=1}^{\lfloor Nr \rfloor} Z'_{j+m_1}\rangle
        -\langle\frac{1}{\sqrt{n}}\sum_{j=1}^{\lfloor Nr \rfloor} Z'_{j+m_1},\sqrt{n}\,\delta_n\rangle.
    \end{align}
    Now, we divide our analysis into three subcases based on the asymptotic behavior of\newline  $\sqrt{n}\,\|\delta_n\|_{\mathcal{H}(K)}$, defined as $\Delta_n$.
    \begin{enumerate}
        \item $\Delta_n\to 0:$ This is equivalent to $\sqrt{n}\delta_n\to \mathbf{0}\in \mathcal{H}(K).$ Then,
        \begin{align*}
            n \langle \hat{\nu}_1-\hat{\nu}_n-\delta_n, \delta_n\rangle \leq \underbrace{\sqrt{n}\left\|\hat{\nu}_1-\hat{\nu}_n-\delta_n\right\|}_{O_{\bP}(1)\text{ because of}~\ref{eqn:CPD_Power_1}}\times \underbrace{\Delta_n}_{\to 0}\overset{\bP}{\to} 0.
        \end{align*}
         Therefore, letting $k=\lfloor Nr \rfloor$, where $0\leq r \leq 1$, from Equations~(\ref{eqn:CPD_Power_1}) and (\ref{eqn:CPD_Power_2}) and Lemma~\ref{lemma:CPD_Power}, we have
        \begin{align*}
            &N^{1 / 2} T_n(k)\\
            &\leadsto \frac{1}{\eta}\left\langle \tilde{W}(\eta)-\tilde{W}(1)+\tilde{W}(1-\eta), \tilde{W}(r(1-2\eta)+\eta)-\tilde{W}(\eta)-r(\tilde{W}(1-\eta)-\tilde{W}(\eta))\right\rangle.
        \end{align*}
        \begin{align*}
            &N V_n(k)\\
            &\leadsto \frac{1}{\eta^2(1-2\eta)}\Bigg(\int_{0}^{r}\left\langle \tilde{W}(\eta)-\tilde{W}(1)+\tilde{W}(1-\eta),\tilde{W}(s(1-2\eta)+\eta)-\tilde{W}(\eta)\right\rangle^2ds\\
            &+\int_{r}^{1}\langle \tilde{W}(\eta)-\tilde{W}(1)+\tilde{W}(1-\eta),\tilde{W}(1-\eta)-\tilde{W}(s(1-2\eta)+\eta)\\
            &-\frac{1-s}{1-r}(\tilde{W}(1-\eta)-\tilde{W}(r(1-2\eta)+\eta))\rangle^2ds\Bigg).
        \end{align*}
        Now, observe that the above two limiting processes are same as the ones we derived under the null. Therefore, following the same lines of argument, we can establish that as $n \to \infty$, $G_n \overset{D}{\to}G.$
        Therefore, $\bP(G_n>G_{1-\alpha})\to \alpha.$
        \item $\Delta_n\to \infty:$ In this case, as $n \to \infty$,
        \begin{align*}
            \frac{1}{\left\|\Delta_n\right\|^2} \left\langle\hat{\nu}_1-\hat{\nu}_n-\delta_n, \sum_{j=1}^{\lfloor Nr \rfloor} Z'_{j+m_1}\right\rangle\leadsto 0.
        \end{align*}
        Also,
        \begin{align*}
            \frac{n}{\Delta_n^2}\left\langle\hat{\nu}_1-\hat{\nu}_n-\delta_n, \delta_n\right\rangle\overset{\bP}{\to}-1.
        \end{align*}
        Now, letting $r_0=\lim_{n \to \infty}\frac{k_0-m_1}{N}$, we obtain that as $n\to \infty$,
        $$\frac{N^{1/2}T_n(k)}{\Delta_n^2}\leadsto (1-2\eta)(r_0\wedge r)((1-r_0)\wedge (1-r)).$$
    Similarly,
    \begin{align*}
        \frac{N}{\Delta_n^4}V_n(k)&\leadsto \int_{0}^{r}\left((1-2\eta)\frac{[(s \wedge r_0)((r-r_0)\wedge (r-s))]\vee 0}{r}\right)^2 ds\\
        &+\int_{r}^{1} \left((1-2\eta)\frac{[((1-s) \wedge (1-r_0))((r_0-r)\wedge (s-r))]\vee 0}{1-r}\right)^2 ds.
    \end{align*}
    In particular, 
    \begin{align*}
        \frac{N^{1/2}T_n(\lfloor N r_0 \rfloor)}{\Delta_n^2}\overset{\bP}{\to} (1-2\eta)r_0(1-r_0)\text{ and } \frac{N}{\Delta_n^4}V_n(\lfloor N r_0 \rfloor)\overset{\bP}{\to} 0.
    \end{align*}
    Therefore,
    $$G_n=\sup_{k=1,2,\cdots,N-1} T_n(k)/V^{1/2}_n(k)\geq T_n(\lfloor N r_0 \rfloor)/V_n(\lfloor N r_0 \rfloor)\overset{\bP}{\to}\infty,$$
    which implies that $\bP(G_n>G_{1-\alpha})\to 1.$
    \item $\Delta_n\to c\in(0,\infty)$ and $\frac{\delta_n}{\left\|\delta_n\right\|}\to \delta_0\in \mathcal{H}(K):$ In this case,
    \begin{align*}
        &\left\langle\hat{\nu}_1-\hat{\nu}_n-\delta_n, \sum_{j=1}^{\lfloor Nr \rfloor} Z'_{j+m_1}\right\rangle \\
        &\leadsto \left\langle \underbrace{\frac{1}{\eta}( \tilde{W}(\eta)-\tilde{W}(1)+\tilde{W}(1-\eta))}_{=b(\eta)}-c\delta_0,\tilde{W}(r(1-2\eta)+\eta)-\tilde{W}(\eta)\right \rangle,
    \end{align*}
    and
    \begin{align*}
        n \langle \hat{\nu}_1-\hat{\nu}_n-\delta_n, \delta_n \rangle \overset{D}{\to} \left\langle b(\eta)-c\delta_0, c\Delta_0\right \rangle.
    \end{align*}
    Therefore, letting $k=\lfloor Nr \rfloor$, we get
    \begin{align*}
        N^{1/2} T_n(k) &\leadsto \left\langle b(\eta)-c\delta_0, \tilde{W}(r(1-2\eta)+\eta)-\tilde{W}(\eta)-r(\tilde{W}(1-\eta)-\tilde{W}(\eta))\right \rangle\\
        &-(1-2\eta)(r_0\wedge r)((1-r_0)\wedge (1-r))\left\langle b(\eta)-c\delta_0, c\delta_0\right \rangle=:T(r,b(\eta)).
        \end{align*}
        \begin{align*}
            &NV_n(k)\\
            &\leadsto \int_{0}^{r} \Bigg(\left \langle b(\eta)-c\delta_0, \tilde{W}(s(1-2\eta)+\eta)-\tilde{W}(\eta)-\frac{s}{r}(\tilde{W}(r(1-2\eta)+\eta)-\tilde{W}(\eta))\right \rangle\\
            &-(1-2\eta)\frac{[(s \wedge r_0)((r-r_0)\wedge (r-s))]\vee 0}{r}\langle b(\eta)-c\delta_0, c\delta_0 \rangle \Bigg)^2 ds\allowdisplaybreaks\\
            &+\int_{r}^{1} \Bigg( \langle b(\eta)-c\delta_0, \tilde{W}(1-\eta)-\tilde{W}(s(1-2\eta)+\eta)\\
            &-\frac{1-s}{1-r}(\tilde{W}(1-\eta)-\tilde{W}(r(1-2\eta)+\eta)) \rangle\\
            &-(1-2\eta)\frac{[((1-s) \wedge (1-r_0))((r_0-r)\wedge (s-r))]\vee 0}{1-r}\langle b(\eta)-c\delta_0, c\delta_0 \rangle \Bigg)^2 ds\\
            &=:V(r, b(\eta)).
        \end{align*}
        Therefore,
        \begin{align*}
            G_n\overset{D}{\to} G^*:=\sup_{r \in [0,1]} T(r, b(\eta))/V^{1/2}(r, b(\eta)).
        \end{align*}
        Now, conditioning on  $b(\eta)-c\delta_0=u$, using the the independent increment property of the Hilbert-valued Brownian motion $W$ as we did in the proof of Theorem 3.3., Lemma~\ref{lemma:distn_inner.prod_BM} and also the fact that $$\frac{1}{\sqrt{1-2\eta}}\left\{B((1-2\eta)t+\eta) -B(\eta)\right\}_{t\in [0,1]}\overset{D}{=}\left\{B(s)\right\}_{s \in [0,1]},$$ we obtain that
        \begin{align*}
            &G^*\mid \{b(\eta)-c\Delta_0=u\}\\
            &\overset{D}{=}\sup_{r \in [0,1]}\frac{\sqrt{\langle Su, u \rangle} (B(r)-rB(1))-(1-2\eta)(r_0\wedge r)((1-r_0)\wedge (1-r))\langle u, c\delta_0 \rangle }
            {\sqrt{\begin{pmatrix}
                \int_{0}^{r}\Big(\sqrt{\langle Su, u \rangle}(B(s)-\frac{s}{r}B(r))\\
                -\sqrt{1-2\eta}\frac{[(s \wedge r_0)((r-r_0)\wedge (r-s))]\vee 0}{r}\langle u, c\delta_0 \rangle\Big)^2 ds\\
                +\int_{r}^{1}\big(\sqrt{\langle Su, u \rangle}(B(1-s)-\frac{1-s}{1-r}B(1-r))\\
                -\sqrt{1-2\eta}\frac{[((1-s) \wedge (1-r_0))((r_0-r)\wedge (s-r))]\vee 0}{r}\langle u, c\delta_0 \rangle\big)^2 ds
            \end{pmatrix}}}\\
            &=\sup_{r \in [0,1]}\frac{ B(r)-rB(1)-\sqrt{\frac{1-2\eta}{\langle Su, u \rangle}}(r_0\wedge r)((1-r_0)\wedge (1-r))\langle u, c\delta_0 \rangle }{\sqrt{\splitdfrac{\int_{0}^{r}\left((B(s)-\frac{s}{r}B(r))-\sqrt{\frac{1-2\eta}{\langle Su, u \rangle}}\frac{[(s \wedge r_0)((r-r_0)\wedge (r-s))]\vee 0}{r}\langle u, c\delta_0 \rangle\right)^2 ds}%
              {\splitdfrac{+\int_{r}^{1}\Big((B(1-s)-\frac{1-s}{1-r}B(1-r))}{-\sqrt{\frac{1-2\eta}{\langle Su, u \rangle}}\frac{[((1-s) \wedge (1-r_0))((r_0-r)\wedge (s-r))]\vee 0}{r}\langle u, c\delta_0 \rangle\Big)^2 ds}}
            }}\\
            &=G^*_{cond}.
        \end{align*}
        This completes the proof.
        \end{enumerate}

\subsection{Proof of Proposition 1.}
Let us define
\[W_{n,i}:=
    \frac{1}{\sqrt{n}}(K(Y_i,\cdot)-\mu_P).\]
where $\{Y_1,Y_2,\cdots,Y_n\}$ comes from a stationary process on $\mathbb{R}^p$ and $P$ is the marginal distribution of $Y_1$. Let us define the two following partial sum processes in $D_{\mathcal{H(K)}}[0,1]$ and $C_{\mathcal{H(K)}}[0,1]$, respectively:
 \[M_n(t):=\sum_{i=1}^{\lfloor nt \rfloor} W_{n,i},\;t\in [0,1];\]
   \[M_n^*(t)=M_n(t)+(nt-\lfloor nt\rfloor)W_{n,\lfloor nt\rfloor+1},\;t\in [0,1].\]  
We first show that $M_n^*\leadsto W$ in $C_H[0,1]$. It suffices to verify the conditions of Lemma~\ref{lemma:fclt_chen_white}.  For a Hilbert valued random element $h_{random}$, we will denote its Hilbert space norm as $\left\|h_{random}\right\|$ and its $L_p$ norm as $\left\|h_{random}\right\|_p$, i.e.,
\[\left\|h_{random}\right\|_p:=(\bE \left\|h_{random}\right\|^p)^{1/p}.\]
First, note that for $1\leq i \leq n$,
\[\left\|W_{n,i}\right\|^2_2=\frac{1}{n}(\bE K(Y_i,Y_i)-\left\|\mu_P\right\|^2)<\infty.\]
 Moreover,
\begin{enumerate}
    \item [(a)]  In our setting, let $V_{n,i}:=Y_i$ for $i=1,2,\cdots,n$. Clearly, $\{W_{n,i}\}$ is $\cL_2$-NED on $\{V_{n,i}\}$ with approximation constants $\mu_m=0$ and $d_{n,i}=0$ for all $m,i$. Moreover, $\{V_{n,i}\}$ is a strong mixing array with $\alpha(m)$ of size $-\frac{r}{r-2}$ with $r=\delta+2$, because $\sum_{m\geq 1} \alpha(m)^{\frac{\delta}{\delta+2}}<\infty.$ Finally, from the definition of $c_{n,i}$ in Lemma~\ref{lemma:fclt_chen_white}, we have
    \[c^2_{n,i}=\max\{d^2_{n,i},\left\|W_{n,i}\right\|^2_2\}=\frac{1}{n}(\bE K(Y_i,Y_i)-\left\|\mu_P\right\|^2)=\frac{1}{n}(\bE K(Y_1,Y_1)-\left\|\mu_P\right\|^2).\]
    Therefore, 
    \begin{align*}
        \sup_{k,J}\,\limsup_{n\to \infty}k^{-1}\sum_{J+1\leq i \leq J+k} (c_{n,i})^2&=\sup_{k,J}\,\limsup_{n\to \infty}k^{-1}\sum_{J+1\leq i \leq J+k} \left\|W_{n,i}\right\|_2^2\\
        &= \frac{1}{n}(\bE K(Y_1,Y_1)-\left\|\mu_P\right\|^2)<\infty.
    \end{align*}
    \item[(b)] Note that when $1\leq i \leq n$,
    \begin{align*}
        \left\|W_{n,i}/c_{n,i}\right\|^2=\frac{\left\|K(Y_i,\cdot)-\mu_P\right\|^2}{\bE K(Y_i,Y_i)-\left\|\mu_P\right\|^2}\overset{D}{=}\frac{\left\|K(Y_1,\cdot)-\mu_P\right\|}{\bE K(Y_1,Y_1)-\left\|\mu_P\right\|^2},
    \end{align*}
    where the equality in distribution follows from stationarity. Therefore, \\$\{\left\|W_{n,i}/c_{n,i}\right\|^2;\;n,i\geq 1\}$ is uniformly integrable.

        \item [(c)] Take any $h\neq 0 \in \mathcal{H}(K)$. First, note that by Cauchy-Schwartz inequality, we have
        \[|h(Y)|=|\langle h, K(Y,\cdot)\rangle|\leq \left\|h\right\|\,\left\|K(Y,\cdot)\right\|=\left\|h\right\|\,\sqrt{K(Y,Y)}.\]
        Raising both sides to the power of $2+\delta$, we have
        \[|h(Y)|^{2+\delta}\leq \left\|h\right\|\, |K(Y,Y)|^{1+\frac{\delta}{2}}.\]
        Since $h\in \mathcal{H}(K)$, its Hilbert space norm is finite and hence, taking expectation and using the monotonicity of $L_r$ norm for $r\geq 1$, we have that \(\left\|h(Y)\right\|_{p}<\infty\) for any $1\leq p\leq 2+\delta.$ Now,
        \begin{align*}
            &\bE\left(\sum_{1\leq i \leq \lfloor nt \rfloor} \langle W_{n,i},h \rangle\right)^2\\
            &=\sum_{1\leq i \leq \lfloor nt \rfloor} \bE \langle W_{n,i},h\rangle^2+2\sum_{1\leq i_1<i_2\leq \lfloor nt \rfloor} \bE \langle W_{n,i_1},h\rangle \langle W_{n,i_2},h\rangle \\
            &=\frac{1}{n}\left(\sum_{1\leq i \leq \lfloor nt \rfloor} \bE (h(Y_i)-\bE h(Y_i))^2+2\sum_{1\leq i_1<i_2\leq \lfloor nt \rfloor} \bE (h(Y_{i_1})-\bE h(Y_{i_1}))(h(Y_{i_2})-\bE h(Y_{i_2})) \right)\\
            &=\frac{1}{n}Var \left(\sum_{1\leq i \leq \lfloor nt \rfloor} (h(Y_i)-\bE h(Y_i))\right),
        \end{align*}
        where the equality in the penultimate line holds because of the reproducing property of the RKHS and also the property of the RKHS mean embedding. Now, observe that by definition, the sequence $\{h(Y_k)\}$ is also strong mixing and its $\alpha$-mixing coefficients are dominated by the ones of $\{Y_k\}$. Therefore, by Lemma~\ref{lemma:dehling_3.2}, it holds that
        \[\bE\left(\sum_{1\leq i \leq \lfloor nt \rfloor} \langle W_{n,i},h \rangle\right)^2=\frac{1}{n}Var \left(\sum_{1\leq i \leq \lfloor nt \rfloor} (h(Y_i)-\bE h(Y_i))\right)\to t\,\sigma^2(h),\]
        where
        \[\sigma^2(h):=Var(h(Y_1))+2\sum_{v\geq 2}Cov(h(Y_1),h(Y_v)).\]
    \item [(d)] \begin{align*}
        \bE \left\|\sum_{i=1}^{n} W_{n,i}\right\|^2
        &=\sum_{i,j=1}^{n}\bE \langle W_{n,i}, W_{n,j}\rangle\allowdisplaybreaks\\
        &\leq \sum_{i,j=1}^{n}|\bE\langle W_{n,i}, W_{n,j}\rangle|\\
        &\leq \sum_{i,j=1}^{n} \frac{15}{n}\alpha(|i-j|)^{\frac{\delta}{\delta+2}}\left\|K(Y_i,\cdot)-\mu_P\right\|_{2+\delta} \left\|K(Y_j,\cdot)-\mu_P\right\|_{2+\delta} \\
        &\lesssim \left\|K(Y_1,\cdot)-\mu_P\right\|_{2+\delta}^2+ \frac{1}{n} \sum_{1\leq i\neq j \leq n}\alpha(|i-j|)^{\frac{\delta}{\delta+2}},
    \end{align*}
    where the penultimate inequality follows from Lemma~\ref{lemma:dehling_3.1}. Now, a simple calculation  further gives us
    \[\bE \left\|\sum_{i=1}^{n} W_{n,i}\right\|^2 \lesssim \left\|K(Y_1,\cdot)-\mu_P\right\|_{2+\delta}^2+2 \sum_{k=1}^{n-1}(1-k/n)\alpha(k)^{\frac{\delta}{\delta+2}}.\]
    Applying the summability of the series $\sum_{k\geq 1}\alpha(k)^{\frac{\delta}{\delta+2}}$ along with Dominated Convergence Theorem, we conclude that the RHS in the above inequality converges to 
    \[\left\|K(Y_1,\cdot)-\mu_P\right\|_{2+\delta}^2+2\,\sum_{k\geq 1}\alpha(k)^{\frac{\delta}{\delta+2}}<\infty.\]
    Therefore, \(\limsup_{n \to \infty} \bE \left\|\sum_{i=1}^{n} W_{n,i}\right\|^2<\infty.\)
    \item [(e)] Continuity of the kernel along with separability of the underlying space (e.g. the Euclidean space) implies separability of the RKHS $\cH(K)$. Thus, $\cH(K)$ indeed has a complete, countable orthonormal basis $\{e_l\}$. Moreover, it's not hard to see that for the NED sequence $\{\langle W_{n,i},e_l\rangle\}$ where $l$ is fixed, one can choose \[c^2_{n,i,l}=\max\{0,\left\|\langle W_{n,i},e_l\rangle\right\|_2^2\}=\frac{1}{n}\bE\left\|e_l(Y_1)-\bE e_l(Y_1)\right\|^2;\] 
    see \cite{wooldridge1988some}. Now, for fixed $n,i,k$, 
    \begin{align*}
        \sum_{l> k} c^2_{n,i,l}&=\frac{1}{n}\sum_{l> k}\bE\left\|e_l(Y_1)-\bE e_l(Y_1)\right\|^2\\
        &=\frac{1}{n}\,\sum_{l> k}\bE\langle K(Y_1,\cdot)-\mu_P, e_l\rangle^2.
    \end{align*}
    Therefore, by Tonelli's Theorem,
    \[\limsup_{n \to \infty}\sum_{l>k}\sum_{1\leq i \leq n}c^2_{n,i,l}=\limsup_{n \to \infty}\sum_{1\leq i \leq n}\sum_{l>k}c_{n,i,l}^2= \sum_{l> k}\bE\langle K(Y_1,\cdot)-\mu_P, e_l\rangle^2.\]
    Now,
    \[\sum_{l\geq 1}\bE\langle K(Y_1,\cdot)-\mu_P, e_l\rangle^2=\bE\sum_{l\geq 1}\langle K(Y_1,\cdot)-\mu_P, e_l\rangle^2=\bE \left\|K(Y_1,\cdot)-\mu_P\right\|^2<\infty,\]
    where the first equality follows from Monotone Convergence Theorem and the final inequality follows from Parseval's identity. Since tail of a summable series converges to $0$, it follows that
    \[\lim_{k \to \infty} \sum_{l> k}\bE\langle K(Y_1,\cdot)-\mu_P, e_l\rangle^2=0.\]
\end{enumerate}
Thus, all the conditions of Lemma~\ref{lemma:fclt_chen_white} are satisfied and thus $M_n^*\leadsto W$ in $C_{\mathcal{H(K)}}[0,1]$, where $W$ is a Brownian motion defined therein. It remains to prove that $M_n\leadsto W$ in $D_{\mathcal{H(K)}}[0,1]$. In the light of Lemma~\ref{lemma:asymp_negligble_process}, it suffices to prove that 
\[\left\|M_n^*-M_n\right\|_{\infty}\overset{\bP}{\to}0.\]
To this end, let us recall that
\[M_n^*(t)=M_n(t)+(nt-\lfloor nt\rfloor)W_{n,\lfloor nt\rfloor+1}.\]
Thus,
\begin{align*}
    \left\|M_n^*-M_n\right\|_{\infty}&=\sup_{t \in [0,1]}\left\|M^*_n(t)-M_n(t)\right\|\leq \max_{1\leq i \leq n}\left\|W_{n,i}\right\|.
\end{align*}
Thus, for any fixed $\epsilon>0$,
\begin{align*}
    \bP( \max_{1\leq i \leq n}\left\|W_{n,i}\right\|>\epsilon)&\leq \sum_{i=1}^n \bP(\left\|W_{n,i}\right\|>\epsilon)\\
    &=n\bP(\left\|W_{n,1}\right\|^{2+\delta}>\epsilon^{2+\delta})\\
    &\leq n\cdot \frac{\bE \left\|K(Y_1,\cdot)-\mu_P\right\|^{2+\delta}}{n^{1+\delta/2}\epsilon^{2+\delta}}\to 0,
\end{align*}
where the first inequality holds by union bound and the last one follows from Markov inequality. Thus, the proof is complete.

        \subsection{Proof of Proposition 2.}
        It suffices to verify the conditions of Lemma~\ref{lemma:fclt_chen_white}. The proof is almost similar to the one of Proposition 1. and we skip the details for conciseness.
        
        \subsection{Proof of Proposition 3.}
        Let us define
$$\Tilde{E}_k:=\begin{pmatrix}
    (K_1(X_{k+m},.)-\mu_1)\otimes (L_{2m+1}((Y_{k},Y_{k+1},\cdots,Y_{k+m},\cdots, Y_{k+2m}),.)-\nu_{2})\\
    (L_1(Y_{k+m},.)-\nu_1)\otimes (K_{2m+1}((X_{k},X_{k+1},\cdots,X_{k+m},\cdots, X_{k+2m}),.)-\mu_{2})
\end{pmatrix}.$$
 Lemma~\ref{lemma:HS_norm_finite} implies that under Assumption 4, $\bE \left\|\Tilde{E}_k\right\|^{2+\delta}<\infty$. Moreover, it is easy to see that $\{\Tilde{E}_i\}_{i \in \mathbb{Z}}$ is $L_2$-NED on the sequence $\{(X_i, Y_i)\}_{i \in \mathbb{Z}}$ with approximation constants $\mu_k=0$ for all $k\geq m$. Therefore, by similar arguments as in Proposition 1., we can conclude that the following FCLT holds:
$$\left\{\frac{1}{\sqrt{n'}}\sum_{i=1}^{\lfloor n'r\rfloor}\left(\Tilde{E}_i-\begin{pmatrix}
    C_{X_0,Y_{-m:m}}\\
    C_{Y_0,X_{-m:m}}
\end{pmatrix}\right)\right\}_{r \in [0,1]}\leadsto W \text{ in } D_{\mathcal{H}(K,L;m)}[0,1].$$
Therefore, it suffices to prove that $\sup_{r \in [0,1]}\left\| \Tilde{D}_{\lfloor n'r \rfloor, n}\right\|_{HS}\overset{\bP}{\to} 0$ as $n\to \infty,$ where
\begin{align*}
    \Tilde{D}_{j,n}:&=\frac{1}{\sqrt{n'}}\sum_{i=1}^{j}\left(\Tilde{E}_i-\begin{pmatrix}
    C_{X_0,Y_{-m:m}}\\
    C_{Y_0,X_{-m:m}}
\end{pmatrix}\right)-\frac{1}{\sqrt{n'}}\sum_{i=1}^{j}\left(E_i-\begin{pmatrix}
    C_{X_0,Y_{-m:m}}\\
    C_{Y_0,X_{-m:m}}
\end{pmatrix}\right)\\
    &=\frac{1}{\sqrt{n'}}\sum_{i=1}^{j}(\Tilde{E}_i-E^{(1)}_i),\;\;j=1,2,\cdots n'.
\end{align*}
Let us define $A_i:=K_1(X_{i+m},.)$, $B_i:=L_{2m+1}(Y_{i:i+2m},.)$, $C_i:=L_1(Y_{i+m},.)$ and \newline \noindent $D_i:=K_{2m+1}(X_{i:i+2m},.)$. Let us also define $\Bar{A}_j=\frac{1}{j}\sum_{i=1}^j A_i$, $\Bar{B}_j=\frac{1}{j}\sum_{i=1}^j B_i$, $\Bar{C}_j=\frac{1}{j}\sum_{i=1}^j C_i$ and $\Bar{D}_j=\frac{1}{j}\sum_{i=1}^j D_i$. Then some routine calculations give us
\begin{align*}
    \Tilde{D}_{j,n}&=\frac{j}{\sqrt{n'}}\left(\Bar{A}_j\otimes (\Bar{B}_{n'}-\bE B_1)+ (\Bar{A}_{n'}-\bE A_1)\otimes \Bar{B}_j+ \bE A_1 \otimes \bE B_1 -\Bar{A}_{n'}\otimes \Bar{B}_{n'}\right)\\
    &+\frac{j}{\sqrt{n'}}\left(\Bar{C}_j\otimes (\Bar{D}_{n'}-\bE D_1)+ (\Bar{C}_{n'}-\bE D_1)\otimes \Bar{C}_j+ \bE C_1 \otimes \bE D_1 -\Bar{C}_{n'}\otimes \Bar{D}_{n'}\right)\allowdisplaybreaks\\
    &=\frac{j}{\sqrt{n'}}\Bigg((\Bar{A}_j-\bE A_1)\otimes (\Bar{B}_{n'}-\bE B_1)+(\Bar{A}_{n'}-\bE A_1)\otimes (\Bar{B}_{j}-\bE B_1)\\
    &-(\Bar{A}_{n'}-\bE A_1)\otimes (\Bar{B}_{n'}-\bE B_1)\Bigg)\\
    &+\frac{j}{\sqrt{n'}}\Bigg((\Bar{C}_j-\bE A_1)\otimes (\Bar{D}_{n'}-\bE D_1)+(\Bar{C}_{n'}-\bE C_1)\otimes (\Bar{D}_{j}-\bE D_1)\\
    &-(\Bar{C}_{n'}-\bE C_1)\otimes (\Bar{D}_{n'}-\bE D_1)\Bigg).
\end{align*}
Now, by Theorem 1 of \cite{sharipov2016sequential}, 
$$\max_{j=1,2,\cdots,n'} \frac{j}{\sqrt{n'}}\left\|\Bar{A}_j-\bE A_1\right\|_{\mathcal{H}(K_1)},\;\;\max_{j=1,2,\cdots,n'} \frac{j}{\sqrt{n'}}\left\|\Bar{B}_j-\bE B_1\right\|_{\mathcal{H}(L_{2m+1})},$$
$$\max_{j=1,2,\cdots,n'} \frac{j}{\sqrt{n'}}\left\|\Bar{C}_j-\bE C_1\right\|_{\mathcal{H}(L_1)},\;\;\max_{j=1,2,\cdots,n'} \frac{j}{\sqrt{n'}}\left\|\Bar{D}_j-\bE D_1\right\|_{\mathcal{H}(K_{2m+1})}$$
converge weakly to the maximum norm of four Hilbert-valued Brownian motions, respectively, and thus they are stochastically bounded. Moreover, by the same theorem $\Bar{A}_{n'}-\bE A_1$, $\Bar{B}_{n'}-\bE B_1$, $\Bar{C}_{n'}-\bE C_1$ and $\Bar{D}_{n'}-\bE D_1$ converge in probability to the zero elements of their corresponding RKHS. Therefore,
\begin{align*}
    &\max_{j=1,2,\cdots,n'}\frac{j}{\sqrt{n'}}\left\|(\Bar{A}_j-\bE A_1)\otimes (\Bar{B}_{n'}-\bE B_1)\right\|_{HS}\\
    &=\max_{j=1,2,\cdots,n'}\frac{j}{\sqrt{n'}}\left\|(\Bar{A}_j-\bE A_1)\right\|_{\mathcal{H}(K_1)}\left\|\Bar{B}_{n'}-\bE B_1\right\|_{\mathcal{H}(L_{2m+1})}\\
    &\overset{\bP}{\to}0,
\end{align*}
where the last line follows from the fact that a norm is a continuous function and thus preserves convergence in probability. Similar argument can be applied to other summands of $\Tilde{D}_{m,n'}$ to complete the proof.
\subsection{Proof of Theorem 3.5.}
Since $\begin{pmatrix}
    C_{X_0,Y_{-m:m}}\\
    C_{Y_0,X_{-m:m}}
\end{pmatrix}=\mathbf{0}$ under $H_0$, the theorem can be proved in a similar fashion of the proofs of Theorem 3.2. and Theorem 3.3.. We omit the details for conciseness. 
\refstepcounter{section}
\section*{Appendix \thesection: Auxilliary Lemmas}\label{sec:appne}
\begin{lemma}\label{lemma:distn_inner.prod_BM}
     Suppose $W$ is an $H$-valued Brownian motion with associated covariance operator $S$. Then, for a fixed $u \in H$, 
     $$\frac{1}{\sqrt{\langle Su, u \rangle}}\times\left\{\langle u, W(t)\rangle\right\}_{t \in [0,1]}\overset{D}{=} \left\{B(t)\right\}_{t \in [0,1]},$$
     where $\left\{B(t)\right\}_{t \in [0,1]}$ is the standard Brownian motion on $[0,1].$
 \end{lemma}
  \begin{proof}[Proof of Lemma~\ref{lemma:distn_inner.prod_BM}]
     For notational simplicity, let us define $\tilde{B}(t):=\frac{1}{\sqrt{\langle Su, u \rangle}}\times \langle u, W(t)\rangle,$ $t \in [0,1].$ 
     \begin{itemize}
         \item [(i)] Clearly, $\tilde{B}(0)=0$ as $W(0)=0.$
         \item[(ii)] Applying the Cauchy-Schwarz inequality, we get  
         \begin{align*}
             \left|\tilde{B}(t_2)-\tilde{B}(t_1)\right|&=\frac{1}{\sqrt{\langle Su, u \rangle}} \left|\langle u, W(t_2)-W(t_1)\right|\\
             &\leq \frac{\left\|u\right\|}{\sqrt{\langle Su, u \rangle}} \left\| W(t_2)-W(t_1)\right\|.\;
         \end{align*}
         
         Since $W\in C_H[0,1]$ almost surely, the above inequality implies $\left\{\tilde{B}(t)\right\}_{t \in [0,1]}\in C[0,1]$ almost surely.
         \item [(iii)] Since $W$ has independent increment property, it's very clear that $\left\{\tilde{B}(t)\right\}_{t \in [0,1]}$ will have the same property as well.
         \item [(iv)] For all $0\leq t<t+s\leq 1$, $W(t+s)-W(t)$ is Gaussian with mean $0$ and covariance operator $sS$, where $S$ is the covariance operator of $W(1)$. Therefore, from Definition \ref{defn:gaussian_hilbert}, it follows that 
         $$\tilde{B}(t+s)-\tilde{B}(t)=\frac{1}{\sqrt{\langle Su, u \rangle}}\cdot \langle u, W(t+s)-W(t)\rangle$$
         has normal distribution on $\mathbb{R}$. Its mean will be
         $$\frac{1}{\sqrt{\langle Su, u \rangle}}\cdot \bE\langle u, W(t+s)-W(t)\rangle=\frac{1}{\sqrt{\langle Su, u \rangle}}\cdot \langle u, \bE(W(t+s)-W(t))\rangle=0.$$
         Thus, the variance of $\tilde{B}(t+s)-\tilde{B}(t)$ will be given by
         \begin{align*}
             \bE(\tilde{B}(t+s)-\tilde{B}(t))^2&=\frac{1}{\langle Su, u \rangle}\cdot \bE\langle u, W(t+s)-W(t)\rangle^2\\
             &=\frac{\langle sSu, u \rangle}{\langle Su, u \rangle}=s \frac{\langle Su, u \rangle}{\langle Su, u \rangle}=s,
         \end{align*}
         where the second and third equalities follow from the definition of covariance operator and homogeneity of inner product.
     \end{itemize}
     Combining $(i)-(iv)$, we conclude the proof.
 \end{proof}

 \begin{lemma}\label{lemma:CPD_Power}
     Under the alternative hypothesis $H'_{1n}$ defined in Section 3.2., it holds that for $k=1, \cdots, N-1$,
\begin{align*}
N^{1 / 2} T_n(k)= & \left\langle\hat{\nu}_1-\hat{\nu}_n-\delta_n, \sum_{j=1}^k Z'_{j+m_1}-\frac{k}{N} \sum_{j=1}^N Z'_{j+m_1}\right\rangle \\
& -\frac{\left(\left(k_0-m_1\right) \wedge k\right)\left(\left(N-k_0+m_1\right) \wedge(N-k)\right)}{N}\left\langle\hat{\nu}_1-\hat{\nu}_n-\delta_n, \delta_n\right\rangle \allowdisplaybreaks\\
N^2 V_n(k)= & \sum_{t=1}^k\left(\left\langle\hat{\nu}_1-\hat{\nu}_n-\delta_n, \sum_{j=1}^t Z'_{j+m_1}-\frac{t}{k} \sum_{j=1}^k Z'_{j+m_1}\right\rangle\right. \\
& \left.-\frac{\left[\left(\left(k_0-m_1\right) \wedge t\right)\left(\left(k-k_0+m_1\right) \wedge(k-t)\right)\right] \vee 0}{k}\left\langle\hat{\nu}_1-\hat{\nu}_n-\delta_n, \delta_n\right\rangle\right)^2 \\
+ & \sum_{t=k+1}^N\Bigg(\left\langle\hat{\nu}_1-\hat{\nu}_n-\delta_n, \sum_{j=t}^N Z'_{j+m_1}-\frac{N-t+1}{N-k} \sum_{j=k+1}^N Z'_{j+m_1}\right\rangle \\
& +\frac{\left[\left(\left(k_0-m_1-k\right) \wedge(t+1-k)\right)\left(\left(N-k_0+m_1\right) \wedge(N-t+1)\right)\right] \vee 0}{N-k}\\
&\times\left\langle\hat{\nu}_1-\hat{\nu}_n-\delta_n, \delta_n\right\rangle\Bigg)^2 .
\end{align*}

\end{lemma}

\begin{proof}[Proof of Lemma~\ref{lemma:CPD_Power}]
    The proof is almost the same as the proof of Proposition 4.1. of \cite{gao2023dimension}, except that we replace all canonical inner products of the Euclidean space with Hilbert space inner products.
\end{proof}

\begin{lemma}\label{lemma:HS_norm_finite}
    Let $H_1$, $H_2$, $H_3$ and $H_4$ be separable Hilbert spaces and $X_i\in H_i$ for $i=1,\cdots,4$. Suppose $\bE \left\|X_i\right\|^{p}<\infty$ for $i=1,\cdots,4$ and some $p\geq 2$. Then, $\bE \left\|Y\right\|^{p/2}<\infty$, where
    $$Y=\begin{pmatrix}
       (X_1-\bE X_1)\otimes  (X_2-\bE X_2)\\
       (X_3-\bE X_3)\otimes  (X_4-\bE X_4)
    \end{pmatrix}.$$
\end{lemma}
\begin{proof}[Proof of Lemma~\ref{lemma:HS_norm_finite}]
From the definition of the inner product of the Hilbert space  $\begin{pmatrix}
       H_1\otimes  H_2\\
       H_3\otimes  H_4
    \end{pmatrix}$ that we mentioned in Section 2.4., it follows that
\begin{align*}
    \left\|Y\right\|^2= \left\|X_1-\bE X_1\right\|^2\cdot \left\|X_2-\bE X_2\right\|^2+\left\|X_3-\bE X_3\right\|^2\cdot \left\|X_4-\bE X_4\right\|^2.
 \end{align*}
 Using the inequality $\sqrt{a^2+b^2}\leq a+b$ for $a,b\geq 0$, we have 
 \begin{align*}
     \left\|Y\right\|\leq \left\|X_1-\bE X_1\right\|\cdot \left\|X_2-\bE X_2\right\|+\left\|X_3-\bE X_3\right\|\cdot \left\|X_4-\bE X_4\right\|.
 \end{align*}
    Therefore,
    \begin{align*}
        \bE \left\|Y\right\|^{p/2}&=\bE\left(\left\|X_1-\bE X_1\right\|\cdot \left\|X_2-\bE X_2\right\|+\left\|X_3-\bE X_3\right\|\cdot \left\|X_4-\bE X_4\right\|\right)^{p/2}\\
        &\leq 2^{p/2-1}\left(\bE\left\|X_1-\bE X_1\right\|^{p/2}\cdot \left\|X_2-\bE X_2\right\|^{p/2}+\bE\left\|X_3-\bE X_3\right\|^{p/2}\cdot \left\|X_4-\bE X_4\right\|^{p/2}\right),
    \end{align*}
    where the inequality follows from $C_r$-inequality. It thus suffices to show that $\bE\left\|X_1-\bE X_1\right\|^{p/2}\cdot \left\|X_2-\bE X_2\right\|^{p/2}<\infty.$ Indeed,
    \begin{align*}
         &\bE (\left\|X_1-\bE X_1\right\|^{p/2}\cdot \left\|X_2-\bE X_2\right\|^{p/2})\\
         &\leq \left(\bE \left\|X_1-\bE X_1\right\|^{p}\cdot \left\|X_2-\bE X_2\right\|^{p}\right)^{1/2}\allowdisplaybreaks\\
         &\leq \left(2^{p-1}(\bE \left\|X_1\right\|^p+\left\|\bE X_1\right\|^p)\cdot 2^{p-1}(\bE \left\|X_2\right\|^p+\left\|\bE X_2\right\|^p)\right)^{1/2}\\
         &\leq \left(2^p \bE \left\|X_1\right\|^p \cdot 2^p \bE \left\|X_2\right\|^p\right)^{1/2}\\
         &=2^p \{\bE\left\|X_1\right\|^p \cdot \bE \left\|X_2\right\|^p\}^{1/2}<\infty,
    \end{align*}
    where the second line follows from Cauchy-Schwartz inequality, the third line follows from $C_r$ inequality and the penultimate inequality follows from Jensen inequality.
\end{proof}

\begin{lemma}\label{lemma:dehling_3.1}[Lemma 3.1. of \cite{dehling1983limit}]
  Let $\mathcal{F}$ and $\mathcal{G}$ be two $\sigma$-fields. Define
$$
\alpha(\mathcal{F},\mathcal{G})=\sup |P(A \cap B)-P(A) P(B)|
$$
the supremum being extended over all $A \in \mathcal{F}$ and $B \in \mathcal{G}$. Let $\xi$ and $\eta$ be random variables with values in a separable Hilbert space $H$ measurable  with respect to $\mathcal{F}$ and $\mathcal{G}$, respectively. If $\xi$ and $\eta$ are essentially bounded then
$$
|\bE\langle \xi, \eta\rangle-\langle \bE \xi,\bE \eta\rangle| \leqq 10 \alpha(\mathfrak{F}, \mathfrak{G})\|\xi\|_{\infty}\|\eta\|_{\infty} .
$$
Here $\|\cdot\|_{\infty}$ denotes the essential supremum with respect to $H$. Moreover, let $r, s, t>1$ with $r^{-1}+s^{-1}+t^{-1}=1$. If $\xi$ and $\eta$ have finite $r$-th and $s$-th moments respectively then
$$
|\bE\langle \xi, \eta\rangle-\langle \bE \xi,\bE \eta\rangle| \leqq 15 \alpha^{1 / t}(\mathcal{F}, \mathcal{G})\|\xi\|_r\|\eta\|_{s} .
$$  
\end{lemma}

\begin{lemma}\label{lemma:dehling_3.2}[Lemma 3.2. of \cite{dehling1983limit}]
     Let $\left\{\xi_v, v \geqq 1\right\}$ be a weakly stationary sequence of $H$-valued random variables centered at expectations and with $(2+\delta)$-th moments uniformly bounded where $0<\delta \leq 1$. If \(\sum_{m\geq 1} \alpha(m)^{\delta/2+\delta}<\infty\), then
$$
\bE\left\|\sum_{v=a+1}^{a+n} \xi_v\right\|^2=\sigma^2 n+o(n),
$$
where
$$
\sigma^2=\bE\left\|\xi_1\right\|^2+2 \sum_{v \geqq 2} \bE\langle\xi_1, \xi_v\rangle.
$$
\end{lemma}
\begin{lemma}\label{lemma:asymp_negligble_process}
    Let $M_n^*$ and $M_n$ be random elements in $C_H[0,1]$ and $D_H[0,1]$ , respectively, defined on a common probability space. Suppose $M_n^*$ converges weakly to $M$ in $C_H[0,1]$. Moreover, suppose the following condition holds:
    \[\left\|M^*_n-M_n\right\|_{\infty}\overset{\bP}{\to}0.\]
    Then, $M_n$ converges weakly to $M$ in $D_H[0,1]$.
\end{lemma}
\begin{proof}
    Since $C_H[0,1]\subset D_H[0,1]$, it follows that $M\in D_H[0,1]$. Note that by the property of the Skorokhod metric,
    \[\rho(M^*_n,M_n)\leq \left\|M^*_n-M_n\right\|_{\infty}\overset{\bP}{\to}0.\]
    Thus, the proof follows from Theorem 18.10 of \cite{van2000asymptotic}.
\end{proof}
\end{appendix}

\end{document}